\documentclass[11pt,a4paper]{article}
\usepackage[a4paper]{geometry}
\usepackage{amssymb,latexsym,amsmath,amsfonts,amsthm}
\usepackage{graphicx}
\usepackage{epstopdf}
\usepackage{tikz}
\usepackage{pgflibraryshapes}
\usetikzlibrary{arrows,decorations.markings}
\usepackage{subfigure}
\usepackage{overpic}

\usepackage{comment,verbatim}
\usepackage{hyperref}

\DeclareMathOperator{\diag}{diag}

\newcommand{\cee}{\mathbb{C}}

\newcommand{\R}{\mathbb{R}}
\newcommand{\C}{\mathbb{C}}

\renewcommand{\Re}{\mathrm{Re}\,}
\renewcommand{\Im}{\mathrm{Im}\,}

\newcommand{\Pe}{\mathrm{Pe}}

\newcommand{\ud}{\,\mathrm{d}}

\newcommand{\wtil}{\widetilde}

\def\supp{\mathop{\mathrm{supp}}\nolimits}

\def\Kcr{\mathop{K^{\mathrm{cr}}}\nolimits}
\def\Ktac{\mathop{K^{\mathrm{tac}}}\nolimits}

\newcommand{\crit}{\textrm{crit}}

\newcommand{\ti}{\tau}

\newtheorem{theorem}{Theorem}[section]
\newtheorem{lemma}[theorem]{Lemma}

\newtheorem{corollary}[theorem]{Corollary}

\newtheorem{rhp}[theorem]{RH problem}

\theoremstyle{definition}
\newtheorem{definition}[theorem]{Definition}

\theoremstyle{remark}

\numberwithin{equation}{section}

\begin{document}

\title{Transitions between critical kernels: from the
tacnode kernel and critical kernel in the two-matrix model to the Pearcey kernel}

\author{Dries Geudens\footnotemark[1] ~~and~ Lun Zhang\footnotemark[1]}
\date{\today}

\maketitle
\renewcommand{\thefootnote}{\fnsymbol{footnote}}
\footnotetext[1]{Department of Mathematics, KU Leuven,
Celestijnenlaan 200B, B-3001 Leuven, Belgium. E-mail:
\{dries.geudens, lun.zhang\}\symbol{'100}wis.kuleuven.be.}

\begin{abstract}
In this paper we study two multicritical correlation kernels and
prove that they converge to the Pearcey kernel in a certain double
scaling limit. The first kernel appears in a model of
non-intersecting Brownian motions at a tacnode. The second arises as
a triple scaling limit of the eigenvalue correlation kernel in the
Hermitian two-matrix model with quartic/quadratic potentials. The
two kernels are different but can be expressed in terms of the same
tacnode Riemann-Hilbert problem. The proof is based on a steepest
descent analysis of this Riemann-Hilbert problem.  A special feature
in the analysis is the introduction of an explicit meromorphic
function on a Riemann surface with specified sheet structure.
\end{abstract}

\setcounter{tocdepth}{2} \tableofcontents

\section{Introduction}

Point processes with determinantal correlation kernels have
attracted a lot of interest over the past few decades due to the
rich mathematical structures behind them and their frequent
occurrences in various random models including random matrix theory,
random growth and tiling problems, etc. \cite{J06,Sosh}. A
fundamental issue in the study of determinantal point processes is
to establish microscopic limits of the correlation kernels, which
often leads to universal results. Besides the well-known canonical
kernels like the sine kernels for bulk universality, the Airy/Bessel
kernels for soft/hard edge universality and so on, some new kernels
were found describing critical behavior of certain random models,
and they are believed to be universal as well. This paper deals with
two such models in which a similar phenomenon occurs. We intend to
retrieve a canonical process (the Pearcey process here) from the
critical phenomenon. Physically, this will lead to descriptions of
phase transitions among different processes. For similar transitions
between canonical processes, we refer to the thesis of Deschout
\cite{Des} and the recent papers \cite{ACV, BC1, BC2}.

\subsection{Non-intersecting Brownian motions at a tacnode}
\begin{figure}[t]
\centering
\begin{overpic}[scale=.3]{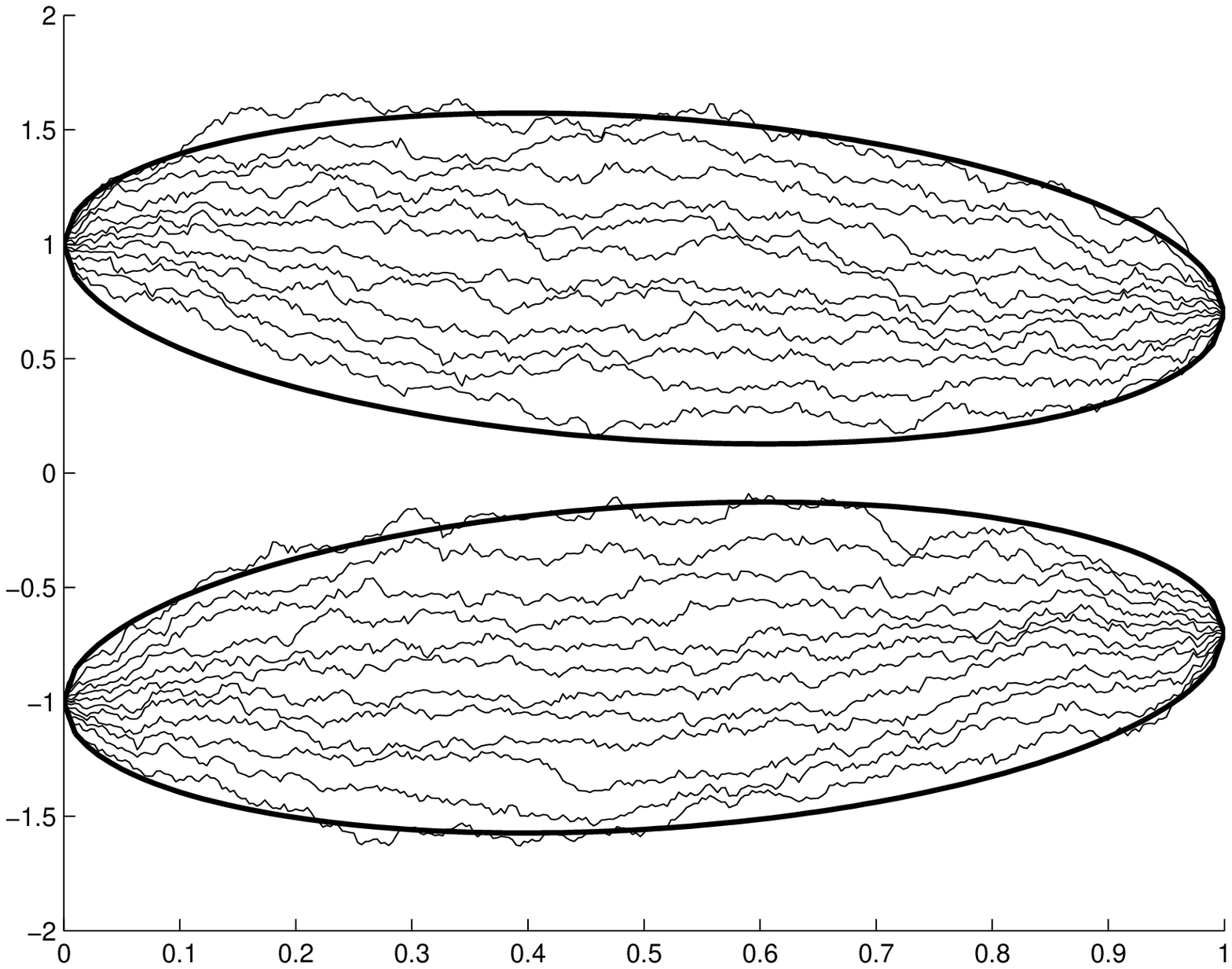}
\small{\put(50,-6){(a)}}
\end{overpic}
\hspace{10mm}
\begin{overpic}[scale=.3]{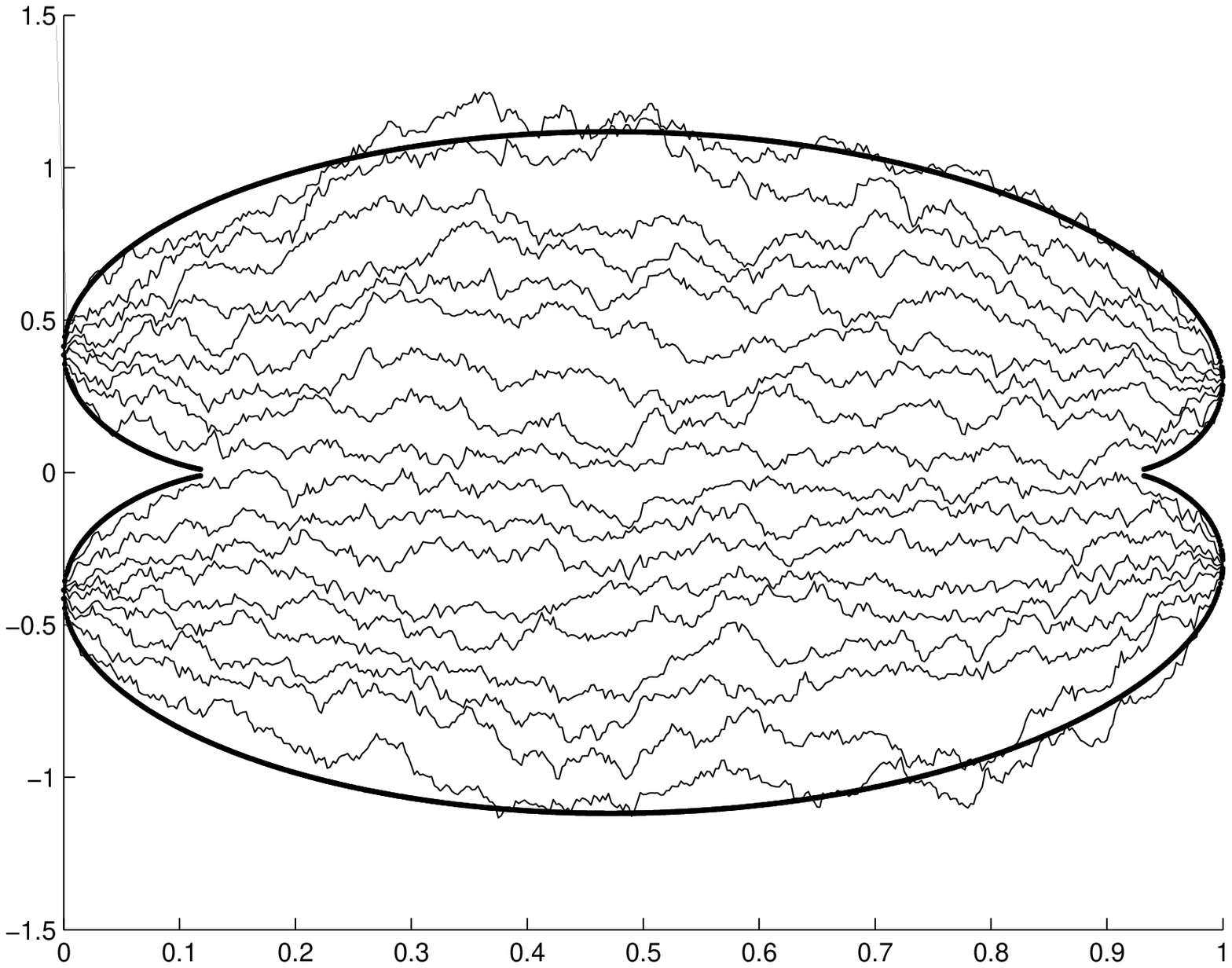}
\small{\put(50,-6){(b)}}
\end{overpic}
\vspace{15mm}
\begin{overpic}[scale=.3]{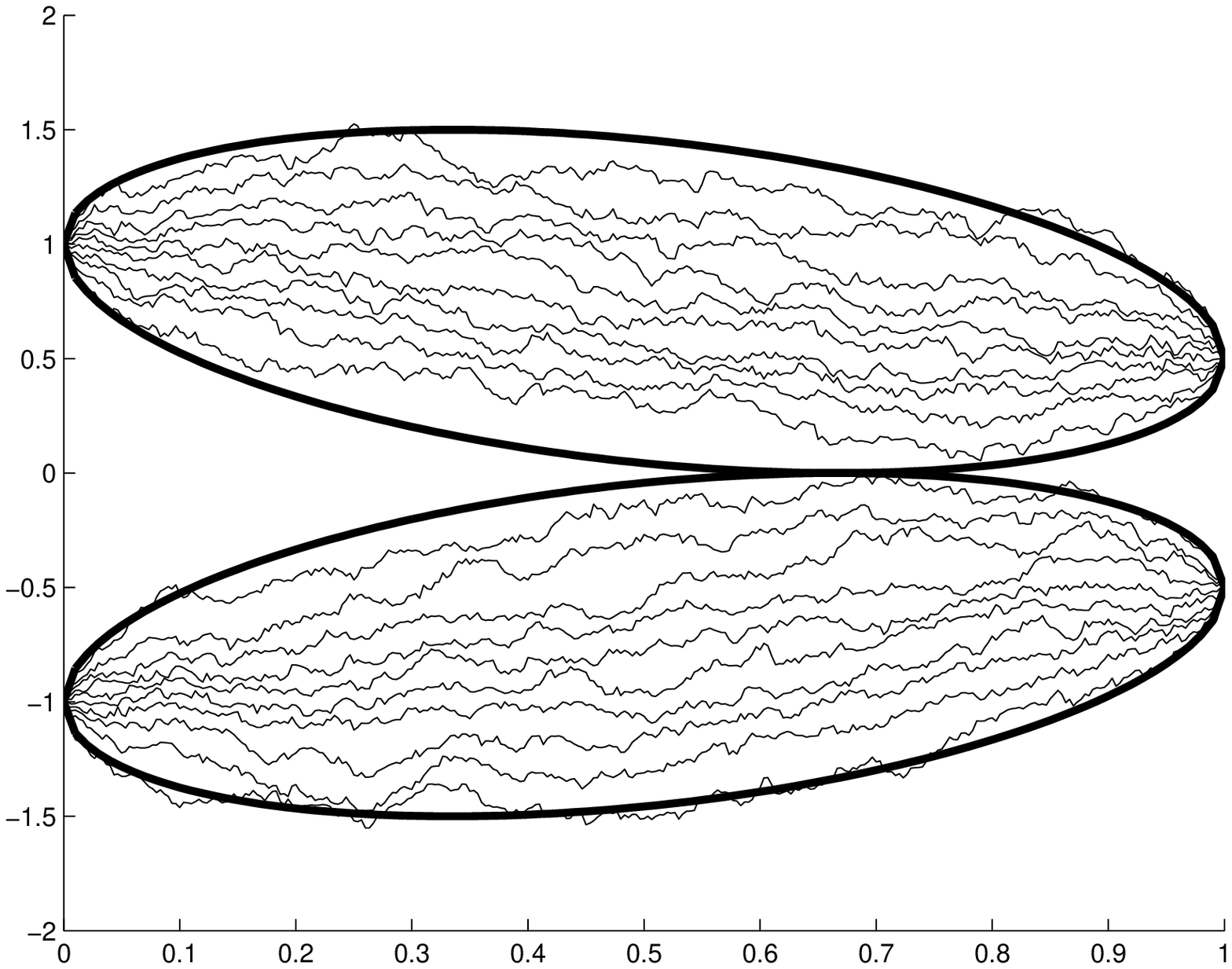}
\small{\put(50,-6){(c)}}
\end{overpic}
\caption{Non-intersecting Brownian
motions with two starting points $\pm\alpha$ and two ending
positions $\pm\beta$ in case of (a) large, (b) small, and (c)
critical separation between the endpoints. Here the horizontal axis
denotes the time, $\ti\in[0,1]$, and for each fixed $\ti$ the
positions of the $n$ non-intersecting Brownian motions at time $\ti$
are denoted on the vertical line through $\ti$. Note that for
$n\to\infty$ the positions of the Brownian motions fill a prescribed
region in the time-space plane, which is bounded by the boldface
lines in the figures. Here we have chosen $n=20,T=1$ in each of the
figures, and (a) $\alpha=1$, $\beta=0.7$, (b) $\alpha=0.4$,
$\beta=0.3$, and (c) $\alpha=1$, $\beta=0.5$, in the cases of large,
small and critical separation, respectively.} \label{fig:3cases}
\end{figure}

As a first model we consider $n$ one-dimensional non-intersecting
Brownian motions with two starting points at time $\tau=0$ and two
ending points at time $\tau=1$. The transition probability density
of the Brownian motions is given by
\begin{equation}\label{transitionprob}
    P_T(\ti,x,y) = \frac{\sqrt{n}}{\sqrt{2\pi \ti T} } \exp\left(- \frac{n}{2\ti T
    }(x-y)^2\right),
\end{equation}
where we interpret $T>0$ as a temperature variable. As the number of
paths tends to infinity, these paths will fill out a certain domain
in the time-space plane. Depending on the locations of the starting
and ending points, the temperature $T$, and the fractions of paths
connecting the topmost and bottommost starting and ending points, we
distinguish three cases, namely, large, small, and critical
separation. This is illustrated in Figure \ref{fig:3cases} which is
taken from \cite{DKV}, see also \cite{DelKui1}.

For each fixed $\ti \in (0,1)$, the positions $x_1,\ldots,x_n$ of
the Brownian paths form a determinantal point process. It is
well-known that the scaling limits of the correlation kernel are
given by the sine kernel in the interior of the domain, the Airy kernel at a typical point of the
boundary, and the Pearcey kernel \cite{BH,BH1} at a cusp
\cite{BK3,TW}. In the critical situation the Brownian
paths fill out two touching ellipses; see Figure
\ref{fig:3cases}(c). The point where the ellipses touch is called the tacnode and the local
particle correlations around this point are described by the tacnode kernel. This tacnode
process was recently studied by different groups of authors using
different techniques. Adler-Ferrari-Van Moerbeke \cite{AFM} resolved
the tacnode problem for non-intersecting random walks (discrete
space and continuous time). Johansson \cite{J} gave an integral
representation of the extended tacnode kernel in the continuous
time-space setting; see also Ferrari and Vet\H o \cite{FV} for
the asymmetric tacnode process. Delvaux-Kuijlaars-Zhang \cite{DKZ} found
an expression of the tacnode kernel in terms of a new $4\times 4$
Riemann-Hilbert (RH) problem. Furthermore, the tacnode process also
appears in a random tiling model called the double Aztec diamond \cite{AJM}.

The critical separation case, see Figure \ref{fig:3cases}(c), can be seen as a limit of the small separation case, in
which the two groups of particles merge at a certain time and
separate again later, see Figure \ref{fig:3cases}(b). In this case
the particles fill out a certain shape in the time-space plane with
two cusp points. At these cusp points the limiting particle density
vanishes with an exponent $1/3$ which indicates that the local
particle correlations are governed by a Pearcey kernel. When we
increase the distance between both starting points and both ending
points keeping the temperature $T$ fixed, or fix the starting and ending
points while lowering the temperature $T$, the cusp singularities
approach each other to form a tacnode. Hence we expect to retrieve the Pearcey kernel from the tacnode
kernel in a certain limit. One aim of this paper is to describe such a transition.

\subsection{Hermitian two-matrix model with quartic/quadratic potentials}
A similar phenomenon occurs in the Hermitian two-matrix
model, which is a probability measure
\begin{equation}\label{twomatrix:DKM}
\frac{1}{Z_n}\exp\left(-n\textrm{Tr}(V(M_1)+W(M_2)-\tau
M_1M_2)\right)\ud M_1\ud M_2,
\end{equation}
defined on the space of pairs $(M_1,M_2)$ of $n\times n$ Hermitian
matrices. Here, $Z_n$ is a normalization constant, $\tau>0$ is the
coupling constant, and $V$ and $W$ are two polynomials which we take
as
\begin{equation}\label{eq:def of VW}
V(x)=\frac{x^2}{2} ,\qquad \textrm{and}\qquad  W(y)=\frac{y^4}{4}+\alpha \frac{y^2}{2},
\end{equation}
for a parameter $\alpha \in \R$.

The two-matrix model can be integrated in terms of biorthogonal
polynomials, see \cite{EMc,BEH1,BEH2,BEH}. The correlation
functions for the eigenvalues of $M_1$ and $M_2$ have a
determinantal structure and admit expressions in terms of these
polynomials \cite{EyM,MS}. Here we are only interested in the
eigenvalues of $M_1$ when averaged over $M_2$, which form a
determinantal point process.

As the parameters $\alpha$ and $\tau$ vary, we encounter phase
transitions. In particular, the point $(-1,1)$ in the
$\alpha \tau$-plane corresponds to a multicritical case. A new
kernel, referred to as the critical kernel in this paper, was recently established by Duits and Geudens \cite{DG} to
describe the local eigenvalue correlations near this point. This critical kernel is
expressed in terms of a RH problem (see RH problem \ref{rhp: tacnode
rhp} below) that is a generalization of the one used in the tacnode
model \cite{DKZ}. Nevertheless, the critical kernel is genuinely different from the tacnode kernel.
Although at the first glance not so clear as in the non-intersecting Brownian motions model, an analysis of the $\alpha\tau$-phase diagram
reveals that this critical kernel can be understood as a limiting case of a Pearcey kernel or a Painlev\'e II kernel \cite{BI,CK}; see
Section \ref{sec:phase diagram in matrix model} for details. In \cite{DG} it was already shown that the critical kernel reduces to
the Painlev\'e II kernel by taking a double scaling limit. The other
aim of this paper is then to retrieve the Pearcey kernel from the
critical kernel, hence, complementing the results in \cite{DG}.

\subsection{The tacnode Riemann-Hilbert problem}
The essential feature allowing us to handle the tacnode kernel and
the critical kernel simultaneously is that both kernels have explicit representations in terms of the following tacnode RH
problem \cite{DKZ,DG}.

\begin{rhp}[Tacnode RH problem] \label{rhp: tacnode rhp}
Fix parameters $r_1,r_2,s_1,s_2$, and $t$. We look for a $4 \times
4$ matrix-valued function $M(z)$ satisfying
\begin{itemize}
\item[\rm (1)] $M$ is analytic for  $z \in \C \setminus \Sigma_M$.
\item[\rm (2)]  For $z\in\Gamma_k$, the limiting values
\[ M_+(z) = \lim_{\substack{\zeta \to z \\\zeta\textrm{ on $+$-side of }\Gamma_k}}M(\zeta), \qquad
   M_-(z) = \lim_{\substack{\zeta \to z \\\zeta\textrm{ on $-$-side of }\Gamma_k}}M(\zeta), \]
exist, where the $+$-side and $-$-side of $\Gamma_k$ are the sides
which lie on the left and right of $\Gamma_k$, respectively, when
traversing $\Gamma_k$ according to its orientation. These limiting
values satisfy the jump relation
\begin{equation}\label{jumps:M}
M_{+}(z) = M_{-}(z)J_k(z),\qquad k=0,\ldots,9.
\end{equation}
\item[\rm (3)] As $z \to \infty$ with $z \in \C \setminus \Sigma_M$ we have
\begin{multline}
M(z)=\left( I+\mathcal O(z^{-1}) \right)B(z) A\\\times \diag \left(
e^{-\psi_2(z)+tz}, e^{-\psi_1(z)-t z}, e^{\psi_2(z)+t
z},e^{\psi_1(z)-t z} \right),
\end{multline}
where
\begin{align}
\psi_1(z)=\frac23 r_1z^{3/2} +2 s_1 z^{1/2}, \quad \psi_2(z)=\frac23
r_2(-z)^{3/2} +2 s_2 (-z)^{1/2} \label{eq: def psi2},
\end{align}
\begin{equation} \label{eq: A}
A=\frac{1}{\sqrt 2} \begin{pmatrix} 1 & 0 & -i & 0 \\ 0 & 1& 0&i \\
-i&0&1&0\\0&i&0&1 \end{pmatrix},
\end{equation}
and
\begin{equation} \label{eq: B}
B(z)=\diag \left((-z)^{-1/4},z^{-1/4},(-z)^{1/4},z^{1/4} \right).
\end{equation}
\item[\rm (4)] $M(z)$ is bounded near $z=0$.
\end{itemize}

The contour $\Sigma_M$ is shown in Figure \ref{fig: contour M} and
consists of 10 rays emanating from the origin. The function $M(z)$
makes constant jumps $J_k$ on each of the rays $\Gamma_k$. These
rays are determined by two angles $\varphi_1$ and $\varphi_2$
satisfying $0<\varphi_1<\varphi_2<\pi/2$. The half-lines $\Gamma_k,$
$k=0,\ldots,9$, are defined by
\begin{align*}
\Gamma_0 = [0,\infty),~\Gamma_1 = [0, e^{i\varphi_1}\infty),~
\Gamma_2 = [0, e^{i\varphi_2}\infty), ~\Gamma_3 =
[0,e^{i(\pi-\varphi_2)}\infty), ~ \Gamma_4=
[0,e^{i(\pi-\varphi_1)}\infty),
\end{align*}
and
\[
\Gamma_{5+k}=-\Gamma_k, \qquad k=0,\ldots,4.
\]
All rays are oriented towards infinity.

The fractional powers are defined with respect to the principal
branch. Hence for example, $z\mapsto (-z)^{3/2}$ is analytic in
$\C\setminus [0,\infty)$ and takes positive values on the negative
part of the real line.

\end{rhp}
\begin{figure}[t]
\centering
\begin{tikzpicture}[scale=.9]
\begin{scope}[decoration={markings,mark= at position 0.5 with {\arrow{stealth}}}]
\draw[postaction={decorate}]      (0,0)--node[near end,
above]{$\Gamma_0$}(4,0) node[right]{$\begin{pmatrix}
0&0&1&0\\0&1&0&0\\-1&0&0&0\\0&0&0&1 \end{pmatrix}$};
\draw[postaction={decorate}]      (0,0)--node[near end,
above]{$\Gamma_1$}(3,1) node[above right]{$\begin{pmatrix}
1&0&0&0\\0&1&0&0\\1&0&1&0\\0&0&0&1 \end{pmatrix}$};
\draw[postaction={decorate}]      (0,0)--node[near end,
right]{$\Gamma_2$}(1.5,2.5) node[above]{$\begin{pmatrix}
1&0&0&0\\-1&1&0&0\\0&0&1&1\\0&0&0&1 \end{pmatrix}$};
\draw[postaction={decorate}]      (0,0)--node[near end,
right]{$\Gamma_3$}(-1.5,2.5) node[above]{$\begin{pmatrix}
1&1&0&0\\0&1&0&0\\0&0&1&0\\0&0&-1&1 \end{pmatrix}$};
\draw[postaction={decorate}]      (0,0)--node[near end,
above]{$\Gamma_4$}(-3,1) node[above left]{$\begin{pmatrix}
1&0&0&0\\0&1&0&0\\0&0&1&0\\0&-1&0&1 \end{pmatrix}$} ;
\draw[postaction={decorate}]      (0,0)--node[near end,
above]{$\Gamma_5$}(-4,0) node[left]{$\begin{pmatrix}
1&0&0&0\\0&0&0&-1\\0&0&1&0\\0&1&0&0 \end{pmatrix}$};
\draw[postaction={decorate}]      (0,0)--node[near end,
above]{$\Gamma_6$}(-3,-1) node[below left]{$\begin{pmatrix}
1&0&0&0\\0&1&0&0\\0&0&1&0\\0&-1&0&1 \end{pmatrix}$} ;
\draw[postaction={decorate}]      (0,0)--node[near end,
right]{$\Gamma_7$}(-1.5,-2.5)node[below]{$\begin{pmatrix}
1&-1&0&0\\0&1&0&0\\0&0&1&0\\0&0&1&1 \end{pmatrix}$};
\draw[postaction={decorate}]      (0,0)--node[near end,
right]{$\Gamma_8$}(1.5,-2.5)node[below]{$\begin{pmatrix}
1&0&0&0\\1&1&0&0\\0&0&1&-1\\0&0&0&1 \end{pmatrix}$};
\draw[postaction={decorate}]      (0,0)--node[near end,
above]{$\Gamma_9$}(3,-1)node[below right]{$\begin{pmatrix}
1&0&0&0\\0&1&0&0\\1&0&1&0\\0&0&0&1 \end{pmatrix}$};
\end{scope}
\end{tikzpicture}
\caption{The jump contour $\Sigma_M$ in the complex $z$-plane and
the constant jump matrices $J_k$ on each of the rays $\Gamma_k$,
$k=0, \ldots,9$.} \label{fig: contour M}
\end{figure}
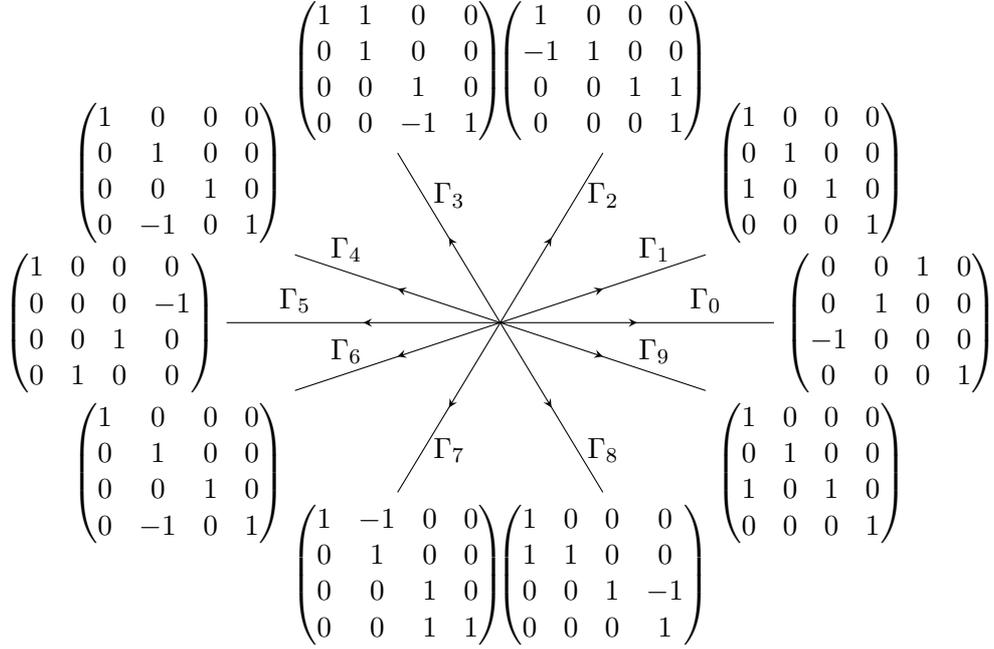

By \cite{DKZ,DG}, this RH problem has a unique  solution for
$r_1=r_2>0$, $s_1=s_2 \in \R$, and $t \in \R$. Moreover the tacnode
RH problem has a remarkable connection with the Hastings-McLeod
solution of the Painlev\'{e} II equation \cite{HML}.

In Sections \ref{sec:tacnode kernel} and \ref{sec:critical kernel} below, we will see that the tacnode RH problem appears in
expressions of the tacnode and critical kernel for the choices of parameters
\begin{align}\label{eq:specialparam}
\begin{cases}
r_1=r_2=1,\\
s_1=s_2=s\in \R,\\
t\in \R.
\end{cases}
\end{align}
We denote the associated solution as $M(z; s, t)$.


\subsection{Outline of the paper}
The rest of this paper is organized as follows. Our results for the
non-intersecting Brownian motions model and the Hermitian two-matrix
model are stated in Sections \ref{sec:tac to Pearcey} and
\ref{sec:critical to Pearcey}, respectively. In both sections, we
first discuss the model from the angle of the phase diagram, which
is particularly helpful in understanding the critical phenomena and
phase transitions that occur. Next we give explicit expressions of
the tacnode and critical kernel in terms of the unique solution to
RH problem \ref{rhp: tacnode rhp} with parameters
\eqref{eq:specialparam}. Our main results are Theorems \ref{th:
tacnode2Pearcey} and \ref{th: Pearcey}, which state that the tacnode
and critical kernel converge to the Pearcey kernel in a certain
double scaling limit. The proofs are given in Section \ref{sec:proof
of thm}. They are based on the Deift-Zhou steepest-descent analysis
\cite{Dei,DKMVZ992,DKMVZ1} of RH problem \ref{rhp: tacnode rhp},
which will be performed in Section \ref{sec: steepest descent}. A
special feature in the analysis is the introduction of an explicit
meromorphic $\lambda$-function on a Riemann surface with specified
sheet structure, which will be the topic of Section \ref{sec:
auxiliary function }. We emphasize that our approach is based on the
steepest-descent analysis of a larger size RH problem as in
\cite{Des} and is different from \cite{ACV,BC1,BC2}.

\section{From the tacnode kernel to the Pearcey kernel}
\label{sec:tac to Pearcey}

\subsection{The phase diagram}
We consider a symmetric version of a model of $n$ 1-dimensional non-intersecting Brownian motions as in \cite{DKV,DelKui1}. The Brownian particles start at points $\pm \alpha$, $\alpha>0$, at time $\ti=0$ and end at points $\pm \beta$, $\beta>0$, at time $\ti=1$. More precisely
we assume that $n$ is even and that $n/2$ of the particles move from the topmost starting point $\alpha$ to the topmost ending
point $\beta$, while the other half of the particles connect the point $-\alpha$ to the point $-\beta$.

For convenience we fix the starting and ending points such that
\begin{equation}\label{eq:fix alpha beta}
2\alpha\beta=1.
\end{equation}
According to \cite{DelKui1}, the three cases of large, small, and
critical separation of the starting and ending points then
correspond to $T<1$, $T>1$, and $T=1$, respectively; see Figure
\ref{fig:3cases} again. Moreover, in case of critical separation,
the coordinate of the tacnode is given by $(\ti_{\crit},0)$, where
\begin{equation}\label{tcrit}
    \ti_{\crit} = \frac{\alpha}{\alpha+\beta}
\end{equation}
is the critical time.

\begin{figure}[t]
\centering
\begin{tikzpicture}
\draw (0,0) node {\includegraphics[scale=0.4]{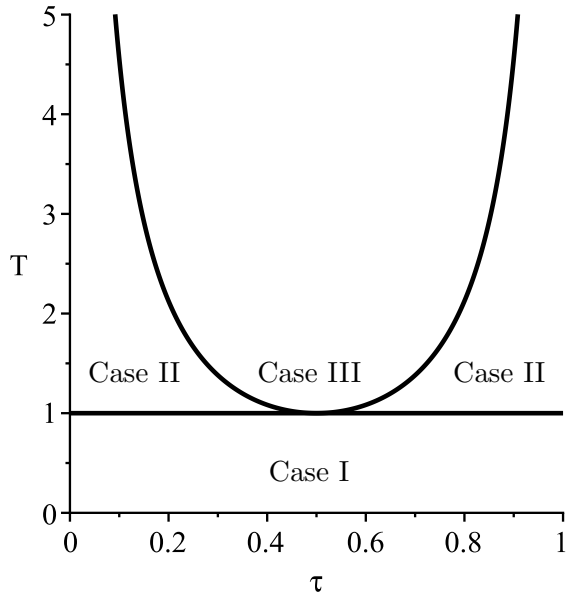}};
\draw (0.3,-2.3) node{Case I}; \draw (0.3,-1) node{Case III}; \draw
(2.8,-1) node{Case II}; \draw (-2,-1) node{Case II};
\end{tikzpicture}
\caption{Phase diagram for $n$ non-intersecting Brownian motions
in the symmetric setting. $n/2$ particles connect $\alpha$ and
$\beta$ while the other particles move from $-\alpha$ to $-\beta$. In
this picture $\alpha=\beta=1/\sqrt 2$. The boldface curves indicate
phase transitions and have equations $T=1$ and
$T=(\ti^2-\ti+1/2)/(\ti(1-\ti))$.} \label{fig: phase diagram B}
\end{figure}

The assumption \eqref{eq:fix alpha beta} allows us to summarize the situation
in a clear phase diagram, shown in Figure \ref{fig: phase diagram B} in the case $\alpha=\beta=1/\sqrt 2$. There are two
curves in the $T \tau$-plane. One is the straight line $T=1$. The
other one is given by the algebraic equation
\begin{equation}\label{eq:critical curve}
T(\ti)=\alpha^2 \frac{1-\ti}{\ti}+\beta^2 \frac{\ti}{1-\ti}.
\end{equation}
These two curves divide the phase plane into three regions, which
we denote by Case I, Case II, and Case III. These three
cases correspond to three generic situations and the two curves
indicate phase transitions. Case I corresponds to temperature $T<1$,
i.e., the case of large separation. In this case, the Brownian paths
remain in two separate groups, and the limiting hull in the
$\tau x$-plane consists of two disjoint ellipses, see Figure
\ref{fig:3cases}(a). The particles are asymptotically distributed on
two intervals according to semicircle laws. In Cases II and III we
have $T>1$ with the paths being distributed on a single interval in
Case III and on two intervals in Case II. In both cases the
densities of the limiting distribution for the positions of the
paths at time $\tau$ are different from semicircle laws. The curve
$T=1$, which corresponds to critical separation in Figure
\ref{fig:3cases}, separates Case I from Case II. The transition occurring at this curve
was studied in \cite{DelKui1} and is related to the homogeneous
Painlev\'e II equation. This phase transition is not visible in the
correlation kernel, but manifests itself in critical limits of the
recurrence coefficients of the associated multiple orthogonal
polynomials. The curve \eqref{eq:critical curve} separates Case II
from Case III. This transition corresponds to the cusps in the
limiting hull in Figure~\ref{fig:3cases}(b), which indicates local behavior described by Pearcey kernels.

The two critical curves touch at the multicritical point
$T=1$, $\ti=\ti_\crit$, which corresponds to the tacnode point in the
critical separation. At this point the limiting particle density
consists of two semicircles touching at the origin. In \cite{DKZ}
the local particle correlations around the origin were studied for fixed temperature $T=1$ and time $\ti=\ti_\crit$, i.e. exactly at the multicritical point $(1,\ti_\crit)$, and varying starting and ending points. Under these assumptions the local scaling limit is expressed in terms of RH problem \ref{rhp: tacnode rhp} with the parameter $t=0$.
In this paper, however, we find it more convenient to fix the starting and ending points $\alpha$
and $\beta$ according to \eqref{eq:fix alpha beta} and scale
$(T,\tau)$ around the multicritical point as follows:
\begin{equation}
\begin{aligned}
\label{eq:scale on t}
\ti&=\ti_{\crit}+Kn^{-1/3}=\frac{\alpha}{\alpha+\beta}+Kn^{-1/3},\\
T&=1+Ln^{-2/3},
\end{aligned}
\end{equation}
where $K$ and $L$ are arbitrary real constants; see also \cite{Del}
for a similar scaling. Then, as $n\to\infty$, the correlation kernel
of non-intersection Brownian motions $K_n(x,y)$ at the tacnode will
converge to the tacnode kernel $\Ktac$ introduced in the next
section, after proper scalings of the space variables.

\subsection{The tacnode kernel}
\label{sec:tacnode kernel}
\begin{definition}
For $u, v \in \mathbb R$, we define the tacnode kernel $\Ktac$ by
\begin{equation} \label{eq:tacnode kernel}
    \Ktac(u,v; s,t)
     = \frac{1}{2\pi i(u-v)} \begin{pmatrix} 0 & 0 & 1 & 1 \end{pmatrix}
    \widehat M^{-1}(v; s,t)
    \widehat M (u; s,t) \begin{pmatrix} 1 \\ 1 \\ 0 \\ 0
    \end{pmatrix},
\end{equation}
where $\widehat M(z;s,t)$ denotes the analytic continuation of the
restriction of $M(z;s,t)$ to the sector around the positive
imaginary axis, bounded by the rays $\Gamma_2$ and $\Gamma_3$.
Recall that $M(z;s,t)$ is the unique solution to RH problem
\ref{rhp: tacnode rhp} with parameters \eqref{eq:specialparam}.
\end{definition}

We can rewrite the kernel in terms of the limiting values
$M_+(u;s,t)$ and $M_+(v;s,t)$ of $M$ on the real line, by using the
jump relations in RH problem \ref{rhp: tacnode rhp}. For example, for $u, v >
0$, we have
\begin{equation}\label{tacnodekernel:pos}
    \Ktac(u,v; s, t)   = \frac{1}{2\pi i(u-v)} \begin{pmatrix} -1 & 0 & 1 & 0 \end{pmatrix}
     M_+^{-1}(v;s,t)  M_+(u;s,t) \begin{pmatrix} 1 \\ 0 \\ 1 \\ 0 \end{pmatrix},
\end{equation}
with a different expression in case $u$ and/or $v$ are negative.

We then have the following results concerning the correlation kernel
of non-intersecting Brownian motions at a tacnode.
\begin{theorem} \label{theorem:kernelpsi}
Consider $n$ nonintersecting Brownian motions on $\R$ with two
starting points $\pm\alpha$, $\alpha>0$ and two ending points $\pm
\beta$, $\beta>0$. The transition probability density is given by
\eqref{transitionprob}. Suppose half of the paths start in $\alpha$ at time $\ti=0$
and end in $\beta$ at time $\ti=1$, while the other half of the paths connect $-\alpha$ and
$-\beta$. If we fix $\alpha$ and $\beta$ such that $2\alpha\beta=1$
and make a double scaling limit of the temperature $T$ and the time $\tau$
according to \eqref{eq:scale on t}, the correlation kernel $K_n$ for
the positions of the particles near the critical time satisfies
\begin{multline} \label{kernel at tacnode}
    \lim_{n \to \infty} \frac{1}{c n^{2/3}}
    K_n \left( \frac{u}{cn^{2/3}},  \frac{v}{cn^{2/3}}\right) \\
    =
    \Ktac(u,v; 2^{-5/3}\left((\alpha + \beta)^4 K^2 -L\right),-2^{-1/3}(\alpha + \beta)^2 K),
\end{multline}
uniformly for $u,v$ in compact subsets of $\mathbb{R}$, where
$c=2^{1/3}(\alpha+\beta)$ and $\Ktac$ is given by \eqref{eq:tacnode
kernel}.
\end{theorem}
\begin{proof}
A version of this theorem with different scaling assumptions was
proved in \cite{DKZ}. The proof of this theorem follows the same
lines. The main difference is in the construction of the local
parametrix around the origin, which is now built out of the solution
to RH problem \ref{rhp: tacnode rhp} for nonzero $t$, whereas in
\cite{DKZ} the version for $t=0$ was used. See also \cite{Del} for a
similar construction in a more involved model.
\end{proof}

\subsection{Double scaling limit: from the tacnode kernel to
the Pearcey kernel}
\label{sec: double scaling limit 1}

The Pearcey kernel $K^\Pe$ is defined as (see \cite{BH,BH1})
\begin{equation}\label{eq: pearcey kernel}
K^\Pe(x,y;\rho)=\frac{p(x)q''(y)-p'(x)q'(y)+p''(x)q(y)-\rho
p(x)q(y)}{x-y},
\end{equation}
with
\begin{equation}\label{eq:pearcey integral}
p(x)=\frac{1}{2\pi}\int_{-\infty}^\infty
e^{-\frac14s^4-\frac{\rho}{2}s^2+isx} \ud s \qquad \text{and} \qquad
q(y)=\frac{1}{2\pi} \int_\Sigma e^{\frac14
t^4+\frac{\rho}{2}t^2+ity} \ud t,
\end{equation}
where $\rho\in\mathbb{R}$. The contour $\Sigma$ consists of the four
rays $\arg t=\pi/4,3\pi/4,5\pi/4,7\pi/4$, where the first and the
third ray are oriented from infinity to zero while the second and
the last ray are oriented outwards. The functions
\eqref{eq:pearcey integral} are called Pearcey integrals
\cite{Pear} and are solutions of the third order differential
equations $p'''(x)=xp(x)+\rho p(x)$ and $q'''(y)=-yq(y)+\rho q(y)$,
respectively.

The Pearcey kernel appears in several random models including the
closing of gap in the Gaussian random matrix model with external
source \cite{BH,BK3,TW}, non-intersecting Brownian motions at
cusps \cite{AM} (see also interpretations in this paper), a
combinatorial model on random partitions \cite{OR}, etc. We mention
that, besides the expression \eqref{eq: pearcey kernel}, the Pearcey
kernel admits a double integral representation as well as a RH
characterization. For the latter one, see Section \ref{sec: RH of
Pearcey kernel}.

This is our first main result.
\begin{theorem} \label{th: tacnode2Pearcey}
Let $\Ktac$ be as defined in \eqref{eq:tacnode kernel} and $K^{\Pe}$
as in \eqref{eq: pearcey kernel}. Then
\begin{equation}
\lim_{a\to -\infty} \frac{1}{\sqrt 2 |a|^{1/4}}\Ktac
\left(\frac{x}{\sqrt 2 |a|^{1/4}}, \frac{y}{\sqrt 2
|a|^{1/4}};-\frac{a^2}{2},|a|\left(1+\frac{\sigma}{2|a|^{3/2}}
\right) \right)=K^{\Pe}(y,x;\sigma), \label{eq: tacnode th 2}
\end{equation}
and
\begin{equation}
\lim_{a\to \infty} \frac{1}{\sqrt 2 a^{1/4}}\Ktac
\left(\frac{x}{\sqrt 2 a^{1/4}}, \frac{y}{\sqrt 2
a^{1/4}};-\frac{a^2}{2},-a\left(1+\frac{\sigma}{2a^{3/2}} \right)
\right)=K^{\Pe}(x,y;\sigma),  \label{eq: tacnode th 1}
\end{equation}
uniformly for $x,y$ in compact sets, where $\sigma\in\mathbb{R}$ is
a parameter.
\end{theorem}

The good scaling of the parameters $s,t$ in the tacnode kernel
$\Ktac$ can be derived from the phase diagram in Figure \ref{fig:
phase diagram B}. Indeed, the phase diagram suggests that it is
possible to retrieve the Pearcey kernel from the tacnode kernel by
letting $\ti$ and $T$ move away from the multicritical point
$(\ti_{\crit},1)$ along the curve \eqref{eq:critical curve} that
separates Case II from Case III. Recall the double scaling
\eqref{eq:scale on t} of $\ti$ and $T$ involving $K$ and $L$. If we
leave the multicritical point to the left/right along the curve, by
\eqref{eq:scale on t}, this amounts to setting
$L=2K^2(\alpha+\beta)^4$ and $\mp K>0$. This implies that the
tacnode kernel appearing in Theorem \ref{theorem:kernelpsi} should
take parameters $s,t$ as
\begin{equation}
s=-a^2/2 \qquad \textrm{and} \qquad t=-a,
\end{equation}
with $\mp a=\mp 2^{-1/3}(\alpha + \beta)^2 K>0$. A further consideration
noting that the Pearcey kernel is a one parameter family invokes us
to take
\begin{equation}
t=-a\left(1+\frac{\sigma}{2|a|^{3/2}} \right),\qquad
\sigma\in\mathbb{R},
\end{equation}
in our setting.

Theorem \ref{th: tacnode2Pearcey} will be proved in Section
\ref{sec:proof of tacnode}. In Sections \ref{sec: auxiliary function }--\ref{sec:proof of thm} we will mainly focus on proving \eqref{eq: tacnode th 2}. In the end we will obtain \eqref{eq: tacnode th 1} from \eqref{eq: tacnode th 2} using a symmetry argument. We next state the results for the
transitions in the Hermitian two-matrix model.

\section{From the critical kernel to the Pearcey kernel}
\label{sec:critical to Pearcey}

\subsection{The phase diagram}
\label{sec:phase diagram in matrix model} As aforementioned, the
eigenvalues of the matrix $M_1$ in the Hermitian two-matrix model
\eqref{twomatrix:DKM} form a determinantal process with correlation
kernel $K_{11}^{(n)}(x,y)$. In this section we will again use a
phase diagram to get more insight into the nature of the critical
phenomena and phase transitions involved.

By the results in \cite{DKM}, it follows that the mean eigenvalue
density for $M_1$ has a limit as $n \to \infty$, i.e., there exists
an absolutely continuous measure $\mu_1$ on $\R$ such that
\begin{align}\label{eq:averagedensitylimit}
\lim_{n\to \infty} \frac{1}{n} K_{11}^{(n)}(x,x)=\frac{\mathrm{d}
\mu_1}{\mathrm{d} x}.
\end{align}
For technical reasons this result was only proved under the
assumption that $n\equiv 0 \mod 3$. A crucial point in the analysis
of \cite{DKM} is that the limiting measure $\mu_1$ can be
characterized by a vector equilibrium problem for vectors of three
measures $(\mu_1,\mu_2,\mu_3)$, see also \cite{HK}. This equilibrium
problem involves a constraint $\sigma_2$ acting on the second
measure $\mu_2$, where $\sigma_2$ is a certain measure on the
imaginary axis. Moreover, the supports of the measures $\mu_1$,
$\sigma_2-\mu_2$, and $\mu_3$ have the following forms
\begin{align*}
    \supp (\mu_1) & =[-\beta_0,-\beta_1] \cup [\beta_1, \beta_0], \\
    \supp(\sigma_2-\mu_2) & =i\R \setminus (-i\beta_2,i\beta_2), \\
    \supp(\mu_3) &= \R \setminus (-\beta_3,\beta_3),
\end{align*}
for some $\beta_0 > \beta_1 \geq 0$, $\beta_2,\beta_3 \geq 0$ that
all depend on the parameters $\alpha \in \mathbb R$ and $\tau > 0$.
Generically, at least one of these numbers is zero, and no two
consecutive ones are zero. This leads to the following
classification of cases as proved in \cite{DGK,DKM}.
\begin{description}
\item[Case I:] $\beta_1=0$, $\beta_2>0$, and $\beta_3=0$.
In this case there are no gaps in the supports of the measures
$\mu_1$ and $\mu_3$ on the real line. The constraint $\sigma_2$ is
active along an interval $[-i\beta_2, i\beta_2]$ on the imaginary
axis.
\item[Case II:] $\beta_1>0$, $\beta_2 > 0$, and $\beta_3=0$.
In Case II there is a gap in the support of $\mu_1$,  but there is no gap in the support of $\mu_3$,
which is again the full real line. The constraint is active
along an interval along the imaginary axis.
\item[Case III:] $\beta_1>0$, $\beta_2 = 0$, and $ \beta_3>0$.
In Case III there is a gap in the supports of $\mu_1$ and $\mu_3$, but the constraint
on the imaginary axis is not active.
\item[Case IV:] $\beta_1=0$, $\beta_2>0$, and $\beta_3>0$.
In this case the measure $\mu_1$ is still supported on one interval. However there
is a gap $(-\beta_3, \beta_3)$ in the support of $\mu_3$. As in Case I, the constraint
$\sigma_2$ is  active along an interval $[-i\beta_2, i\beta_2]$ on the imaginary axis.
\end{description}

\begin{figure}[t]
\begin{center}
\begin{tikzpicture}[scale=1.1]
\draw[->](0,0)--(0,4.25) node[above]{$\tau$};
\draw[->](-4,0)--(4.25,0) node[right]{$\alpha$}; \draw[help lines]
(-1,0)--(-1,1)--(0,1); \draw[very thick,rotate around={-90:(-2,0)}]
(-2,0) parabola (-4.5,6.25)   node[above]{$\tau=\sqrt{\alpha+2}$};
\draw[very thick,dashed,rotate around={-90:(-2,0)},color=white]
(-2,0) parabola (-3,1) ; \draw[very thick,dashed] (-1,1)..controls
(0,1.5) and (-0.2,3).. (-0.1,4); \draw[very thick]
(-1,1)..controls (-2,0.5) and (-3,0.2).. (-4,0.1)
node[above]{$\tau=\sqrt{-\frac{1}{\alpha}}$}; \filldraw  (-1,1)
circle (1pt); \draw[very thick,->] (-1,1)--(-0.2,1.4)
node[above]{$a$}; \draw[very thick,->] (-1,1)--(-1.4,1.8)
node[right]{$b$}; \draw (0.1,1)
node[font=\footnotesize,right]{$1$}--(-0.1,1); \draw
(-1,0.1)--(-1,-0.1) node[font=\footnotesize,below]{$-1$}; \draw
(-2,0.1)--(-2,-0.1) node[font=\footnotesize,below]{$-2$}; \draw
(0.1,1.43) node [font=\footnotesize,right]{$\sqrt 2$}--(-0.1,1.43);
\draw[very thick] (2,0.8) node[fill=white]{Case I}
                  (-2.5,0.2) node{Case IV}
                  (-2,2) node[fill=white]{Case III}
                  (1,3) node[fill=white]{Case II};
\end{tikzpicture}
\end{center}
\caption{The phase diagram in the $\alpha\tau$-plane: the
critical curves $\tau=\sqrt{\alpha+2}$ and
$\tau=\sqrt{-\frac{1}{\alpha}}$ separate the four cases. The cases
are distinguished by the fact whether $0$ is in the support of the
measures $\mu_1$, $\sigma_2-\mu_2$, and $\mu_3$, or not.}
\label{fig: phase diagram}
\end{figure}
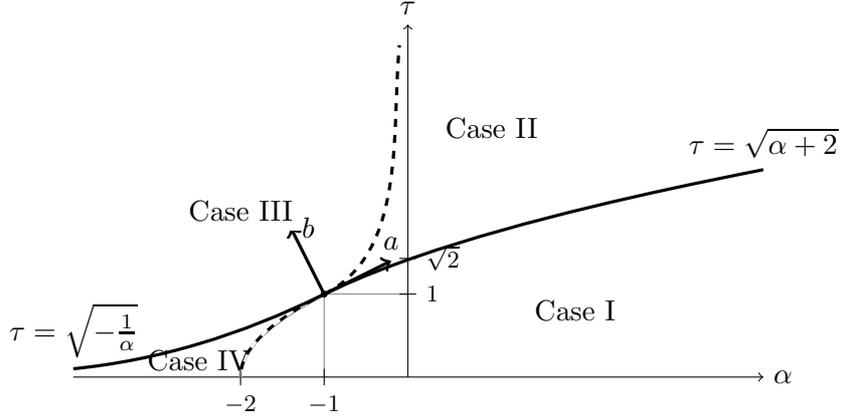

The correspondences between the values of $(\alpha,\tau)$ and these
four cases are illustrated in the phase diagram shown in Figure \ref{fig: phase diagram}. The different cases are separated by two curves with
equations
\begin{equation}
    \tau = \sqrt{\alpha + 2}, \qquad -2 \leq \alpha < \infty,
\end{equation}
and
\begin{equation}\label{eq:curve II}
    \tau = \sqrt{- \frac{1}{\alpha}}, \qquad -\infty < \alpha < 0,
\end{equation}
respectively. On these curves two of the numbers $\beta_1, \beta_2,$
and $\beta_3$ are equal to zero. For example, on the curve
\eqref{eq:curve II} that separates Case III from Case IV, we have
$\beta_1=\beta_2=0$, while $\beta_3>0$. Finally, note the
multi-critical point $(\alpha,\tau)=(-1,1)$ in the phase diagram,
where $\beta_1=\beta_2=\beta_3 = 0$. All four cases come together at
this point in the $\alpha\tau$-plane.

As long as we consider points $(\alpha,\tau)$ that are not on the
curves, the local eigenvalue correlations are governed by the sine
kernel in the bulk of the spectrum and the Airy kernel at the edge
of the spectrum. Critical phenomena occur at the curves that
separate the different cases. When we cross the line
$\tau=\sqrt{\alpha+2}$ for $\alpha>-1$, the support of $\mu_1$ turns
from two intervals into one interval. On this critical curve, the
intervals meet at the origin and  there the density  vanishes
quadratically. This indicates that in a double scaling limit the
correlation kernel converges to the Painlev\'e II kernel that arises
at the closing of a gap in unitary ensembles \cite{BI,CK,CKV}. When
crossing the line $\tau=\sqrt{-1/\alpha}$ for $\alpha<-1$, we again
have a transition of two intervals merging at the origin. However,
due to the fact that $\supp(\sigma_2-\mu_2)$ also closes
simultaneously, the vanishing at the origin occurs with an exponent
${1/3}$. This indicates that the local correlations are governed by
the Pearcey kernel. The other transitions, represented by the dashed
lines in Figure \ref{fig: phase diagram}, do not concern $\mu_1$.
They take place on the non-physical sheets of the spectral curve
and, therefore, they do not influence the local correlations of the
eigenvalues of $M_1$, which are again described by the sine and Airy
kernels.

The limiting process near the multicritical point
$(\alpha,\tau)=(-1,1)$, where there is a simultaneous transition in
the supports of all three measures $\mu_1$, $\sigma-\mu_2$ and
$\mu_3$ is the subject of the recent paper \cite{DG}. Here the
limiting eigenvalue density vanishes like a square root at the
interior of the support, which is similar to the situation in
non-intersecting Brownian motions at a tacnode. By taking triple
scaling limits such that $\alpha$ and $\tau$ are made dependent on
$n$ as
\begin{equation} \label{eq: scaling t tau}
\begin{pmatrix} \alpha \\ \tau \end{pmatrix} = \begin{pmatrix} -1 \\ 1 \end{pmatrix}
+ a n^{-1/3}\begin{pmatrix} 2 \\ 1 \end{pmatrix} + b n^{-2/3}\begin{pmatrix} -1 \\ 2 \end{pmatrix},
\end{equation}
for  $a,b\in \R$, the large $n$ limit of the correlation kernel
$K_{11}^{(n)}(x,y)$ is given by the critical kernel $\Kcr$ that we
will introduce in the next section. Note that in \eqref{eq: scaling
t tau} the vectors $\begin{pmatrix} 2 & 1
\end{pmatrix}^T$ and $\begin{pmatrix} -1 & 2 \end{pmatrix}^T$ are
respectively tangent and normal to both critical curves in the point
$(\alpha,\tau)=(-1,1)$.

\subsection{The critical kernel}
\label{sec:critical kernel}

\begin{definition}
For $u,v\in\mathbb{R}$, we define the critical kernel $\Kcr$ by
\begin{align}\label{eq:defKcr}
\Kcr(u,v;s,t)= \frac{1}{2 \pi i (u-v)} \begin{pmatrix} -1 & 1 & 0 &
0
\end{pmatrix}   M(iu;s,t)^{-1}    M(iv;s,t)
\begin{pmatrix} 1\\1\\0\\0 \end{pmatrix}.
\end{align}
\end{definition}

The following theorem is the main result in \cite{DG}.
\begin{theorem}\label{thm:critical kernel}\cite[Theorem 2.3]{DG}
Let $K_{11}^{(n)}$ be the kernel describing the eigenvalues of $M_1$
when averaged over $M_2$ in the Hermitian two-matrix model
\eqref{twomatrix:DKM} with $V$ and $W$ as in \eqref{eq:def of VW}.
Assume $\alpha$ and $\tau$ depend on $n$ as in \eqref{eq: scaling t
tau}. Then for $n\to \infty$ and $n\equiv 0 \mod 6$, we have
\begin{equation*}
\lim_{n \to \infty}
\frac{1}{n^{2/3}}K_{11}^{(n)}\left(\frac{u}{n^{2/3}},\frac{v}{n^{2/3}}\right)
=\Kcr\left(u,v;\tfrac14 (a^2-5b),-a\right),
\end{equation*}
uniformly for $u,v$ in compact subsets of $\R$, where $\Kcr$ is
defined in \eqref{eq:defKcr}.
\end{theorem}

\subsection{Double scaling limit: from the critical kernel to the Pearcey kernel}
\label{sec: double scaling limit 2}

Our second main result is the following.
\begin{theorem} \label{th: Pearcey}
Let $\Kcr$ be as defined in \eqref{eq:defKcr} and $K^\Pe$ as in
\eqref{eq: pearcey kernel}. Then
\[
\lim_{a\to -\infty} \frac{1}{\sqrt 2 |a|^{1/4}}\Kcr
\left(\frac{x}{\sqrt 2 |a|^{1/4}}, \frac{y}{\sqrt 2
|a|^{1/4}};-\tfrac12
a^2,|a|\left(1-\frac{\sigma}{2|a|^{3/2}}\right)\right)=K^{\Pe}(x,y;\sigma),
\]
uniformly for $x,y$ in compact sets, where $\sigma\in\mathbb{R}$ is
a parameter.
\end{theorem}


The scaling of the parameters $s,t$ in the critical kernel $\Kcr$
is again motivated by the phase diagram in Figure \ref{fig: phase
diagram}. This phase diagram suggests that it is
possible to retrieve the Pearcey kernel from the critical kernel by
letting $t$ and $\alpha$ move away from the multicritical point
along the curve $\tau=\sqrt{-1/\alpha}$, $\alpha<-1$. This amounts
to setting $b=3a^2/5$ and $a<0$. Hence the critical kernel in
Theorem \ref{thm:critical kernel} takes parameters $s=-a^2/2$ and
$t=|a|$. To obtain the limit in the most general form, we perform
the scaling for $t$ as
\begin{equation} \label{eq: def s t}
t=|a|\left(1-\frac{\sigma}{2|a|^{3/2}} \right), \qquad \sigma \in
\R.
\end{equation}
We also note that if one lets $t$ and $\alpha$ move away from the
multicritical point along the curve $\tau = \sqrt{\alpha + 2}$,
$\alpha>-1$, the critical kernel will tend to the Painlev\'e II
kernel; see \cite{DG} for more details.

Theorem \ref{th: Pearcey} will be proved in Section \ref{sec:proof
of critical}.


\section{Meromorphic $\lambda$-functions on a Riemann surface}
\label{sec: auxiliary function }

In this section, we introduce some auxiliary functions and study
their properties. The aim is to construct the so-called
$\lambda$-functions, of which the analytic continuation defines a
meromorphic function on a Riemann surface with specified sheet
structure. The $\lambda$-functions have desired behavior around each
branch point, and we shall make use of them in the normalization
step of our steepest descent analysis in Section \ref{sec: M3 to
M4}.

Throughout this section, unless specified differently, we shall take
the principal cut for all fractional powers.

\subsection{A four-sheeted Riemann surface and the $w$-functions}
\label{sec:w function}
We introduce a four-sheeted Riemann surface $\mathcal R$ with sheets
\begin{align*}
\mathcal R_1 &= \C \setminus [0,\infty), & \mathcal R_3 &= \C \setminus \left([-ic,ic] \cup [0,\infty)\right), \\
\mathcal R_2 &= \C \setminus (-\infty,0],& \mathcal R_4 &= \C \setminus ([-ic,ic] \cup (-\infty,0]) ,
\end{align*}
where
\begin{equation} \label{eq: c}
c=\frac{3 \sqrt 3}{16}\gamma^{-3/2},
\end{equation}
and $\gamma>0$ is a parameter that we will specify later.

We connect the sheets $\mathcal R_j$, $j=1,2,3,4$, to each other in
the usual crosswise manner along the cuts $[-ic,ic]$, $[0,\infty)$,
and $(-\infty,0]$. More precisely, $\mathcal R_1$ is connected to
$\mathcal R_3$ along the cut $[0,\infty)$, $\mathcal R_2$ is
connected to $\mathcal R_4$ along the cut $(-\infty,0]$, and
$\mathcal R_3$ is connected to $\mathcal R_4$ via $[-ic,ic]$.
Moreover, the Riemann surface is compactified by adding two points
at infinity. The first point $\infty_1$ is added to the first and
the third sheet, while the second point $\infty_2$ connects the
other two sheets. This Riemann surface $\mathcal R$ has genus zero
and is shown in Figure \ref{fig: Riemann surface}.

We intend to find functions $\lambda_j$ on these sheets, such that
each $\lambda_j$ is analytic on $\mathcal R_j$ and admits an
analytic continuation across the cuts. Later we will impose extra
conditions on the poles and zeros of this function. For the moment,
we content ourselves constructing an elementary function $w$ that is
meromorphic on $\mathcal R$.

\begin{figure}[t]
\centering
\begin{tikzpicture}[scale=0.7]
\draw (0,0)--(6,0)--(8,1)--(2,1)--cycle
      (0,-2)--(6,-2)--(8,-1)--(2,-1)--cycle
      (0,-4)--(6,-4)--(8,-3)--(2,-3)--cycle
      (0,-6)--(6,-6)--(8,-5)--(2,-5)--cycle;
\draw[very thick] (4,0.5)--(7,0.5)
                  (4,-3.5)--(7,-3.5)
                  (1,-1.5)--(4,-1.5)
                  (1,-5.5)--(4,-5.5)
                  (3.5,-3.75)--(4.5,-3.25)
                  (3.5,-5.75)--(4.5,-5.25);
\draw (8,0.5) node{$\mathcal R_1$}
      (8,-1.5) node{$\mathcal R_2$}
      (8,-3.5) node{$\mathcal R_3$}
      (8,-5.5) node{$\mathcal R_4$};
\end{tikzpicture}
\caption{Riemann surface $\mathcal R$.}
\label{fig: Riemann surface}
\end{figure}
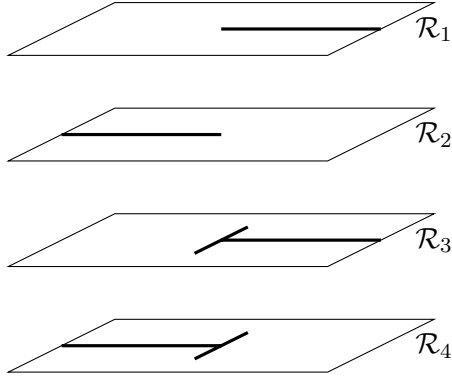

Consider the algebraic equation
\begin{equation} \label{eq: alg eq w}
z=\frac{w}{(\gamma^3 w^2-1)^2}.
\end{equation}
It turns out this equation defines a meromorphic function on the
Riemann surface $\mathcal R$. This is made precise in the following lemma.
\begin{lemma} \label{lemma: w}
There exist functions $w_j(z)$, $j=1,2,3,4,$ that solve the equation
\eqref{eq: alg eq w} and satisfy the following
conditions.
\begin{itemize}
\item[\rm (a)] $w_j(z)$ is analytic on $\mathcal R_j$, $j=1,2,3,4$, and
\begin{align*}
w_{1,\pm}(x) &= w_{3,\mp}(x), && x \in (0,\infty), \\
w_{2,\pm}(x) &= w_{4,\mp}(x), && x \in (-\infty,0), \\
w_{3,\pm}(z) &= w_{4,\mp}(z), && z \in (-ic,0) \cup (0,ic).
\end{align*}
Here, we orient $(-\infty,0)$ and $(0,\infty)$ from left to right,
and $(-ic,0) \cup (0,ic)$ from bottom to top. Hence, the function
$\bigcup_{j=1}^4 \mathcal R_j \to \C: \mathcal R_j \ni z \mapsto
w_j(z)$ has an analytic continuation to a meromorphic function $w:
\mathcal R \to \overline \C$. This function is a bijection.

\item[\rm (b)] $w$ satisfies the symmetry properties
\begin{align*}
w_j(\overline z) &= \overline{w_j(z)},  & z&\in\mathcal R_j,\qquad j=1,2,3,4, \\
w_1(-z)&=-w_2(z), & z&\in \C\setminus (-\infty,0],
\\
w_3(-z)&=-w_4(z), & z&\in \C\setminus \left((-\infty,0]\cup[-ic,
ic]\right).
\end{align*}

\item[\rm (c)] As $z \to \infty$ within $\C\setminus (-\infty,0]$ we have that
\begin{multline*}
w_2(z) =\gamma^{-3/2}+\tfrac12
\gamma^{-9/4}z^{-1/2}-\tfrac{1}{64}\gamma^{-15/4}z^{-3/2}
\\+\tfrac{1}{128}\gamma^{-9/2}z^{-2}-
\tfrac{9}{4096}\gamma^{-21/4}z^{-5/2} +\mathcal O (z^{-7/2}),
\end{multline*}
\begin{multline*}
w_4(z) =\gamma^{-3/2}-\tfrac12
\gamma^{-9/4}z^{-1/2}+\tfrac{1}{64}\gamma^{-15/4}z^{-3/2}
\\+\tfrac{1}{128}\gamma^{-9/2}z^{-2}+
\tfrac{9}{4096}\gamma^{-21/4}z^{-5/2} +\mathcal O (z^{-7/2}).
\end{multline*}

\item[\rm (d)] As $z \to 0$ we have that
\begin{align*}
w_2(z) &=\gamma^{-2}z^{-1/3}+\tfrac23 \gamma^{-1} z^{1/3}-\tfrac13 z
+ \tfrac{28}{81}\gamma z^{5/3}+\mathcal O (z^{7/3}),
\qquad \qquad \text{in } \C \setminus (-\infty,0], \\
w_4(z)& =\begin{cases}   z  -2\gamma^3z^3+9\gamma^6z^5+\mathcal O (z^{7}),  & \text{in }I \cup IV, \\
                         -\omega^2 \gamma^{-2}(-z)^{-1/3}-\tfrac23 \omega
                         \gamma^{-1} (-z)^{1/3}-\tfrac13 z - \tfrac{28}{81} \omega^2 \gamma (-z)^{5/3}+\mathcal O (z^{7/3}),  & \text{in }II, \\
                         \omega^2 \gamma^{-2}z^{-1/3}+\tfrac23 \omega
                         \gamma^{-1} z^{1/3}-\tfrac13 z + \tfrac{28}{81} \omega^2 \gamma z^{5/3}+\mathcal O (z^{7/3}),  & \text{in }III. \end{cases}
\end{align*}
Here, I, II, III and IV stand for the four open quadrants in the complex plane and $\omega=e^{2\pi i/3}$.

\item[\rm (e)] As $z \to ic$ we have that
\[
w_{3}(z)= \frac{i}{\sqrt 3}\gamma^{-3/2}+8\cdot3^{-7/4}e^{3\pi
i/4}\gamma^{-3/4}(z-ic)^{1/2}+\mathcal O \left( z-ic\right),
\]
where we take the branch cut of $(z-ic)^{1/2}$ along $(-i\infty ,
ic)$ and $-\pi/2< \arg (z-ic)<\frac{3 \pi}{2}$.

\item[\rm (f)] We have
$$\displaystyle w_{3,+}\left((0,ic)\right)=\left(0,\tfrac{i}{\sqrt 3}\gamma^{-3/2}\right),
\qquad  \qquad w_{3,-}\left((0,ic)\right)=\left(i
\infty,\tfrac{i}{\sqrt 3}\gamma^{-3/2}\right).$$
\end{itemize}
\end{lemma}
Note that the complete behavior of $w_j$, $j=1,2,3,4$, around the
origin, $\infty$, and $\pm ic$ can be obtained from the symmetry
conditions in (b).
\begin{proof}
We consider the algebraic equation defined by
\begin{equation}\label{eq: alg eq eta}
(\widehat w^2+\gamma^3)^2=z\widehat w^3,
\end{equation}
where again $\gamma$ is a positive parameter as in \eqref{eq: alg eq
w}. For any solution $\widehat w(z)$ of
\eqref{eq: alg eq eta}, the function
$$w(z)=i\left(\widehat w(i/z)\right)^{-1}$$ satisfies equation \eqref{eq:
alg eq w}. Our strategy is then to construct $w_j$ from the
solutions of \eqref{eq: alg eq eta} via the above transformation.
Fortunately, the equation \eqref{eq: alg eq eta} was well-studied in
\cite[Section 3]{DG}. For each $z\in \C$, it was shown that the
equation \eqref{eq: alg eq eta} admits four solutions, denoted by
$\widehat w_j(z)$, $j=1,2,3,4$, such that
\[
|\widehat w_1(z)|\geq|\widehat w_2(z)|\geq|\widehat w_3(z)|\geq|\widehat w_4(z)|.
\]
These functions can be interpreted as a meromorphic function on
another four-sheeted Riemann surface $\hat{\mathcal { R}}$ with
sheets
\begin{align*}
\hat{\mathcal {R}}_1 &= \C \setminus [-\hat{c},\hat{c}], &
\hat{\mathcal {R}}_3 &= \C \setminus (\mathbb R \cup i\mathbb R), \\
\hat{\mathcal {R}}_2 &= \C \setminus \left([-\hat{c},\hat{c}]\cup i\mathbb R\right), &
\hat{\mathcal {R}}_4 &= \C \setminus \mathbb R ,
\end{align*}
where
\[
\hat{c}=1/c=\frac{16}{3 \sqrt 3}\gamma^{3/2}.
\]
These four sheets are connected in the usual crosswise manner and the Riemann surface is compactified by
adding two points at infinity: one on the first sheet, and another to the three remaining sheets. Each $\widehat w_j(z)$ is
analytic on $\hat{\mathcal {R}}_j$ and the function $\bigcup_{j=1}^4
\hat{\mathcal {R}}_j \to \C: \hat{\mathcal {R}}_j \ni z \mapsto
\widehat w_j(z)$ has an analytic continuation to a meromorphic function
$\widehat w: \hat{\mathcal {R}} \to \overline \C$. This fact invokes us to
define the following functions $w_j(z)$:
\begin{align*}
w_1(z) &= \begin{cases} i \left( \widehat w_3(i/z)\right)^{-1}, & \qquad\text{in } I \cup IV, \\
                          i \left( \widehat w_4(i/z)\right)^{-1}, & \qquad \text{in } II \cup III, \end{cases}&
w_2(z) &= \begin{cases} i \left( \widehat w_4(i/z)\right)^{-1}, & \qquad\text{in } I \cup IV, \\
                          i \left( \widehat w_3(i/z)\right)^{-1}, & \qquad\text{in } II \cup III, \end{cases}\\
w_3(z) &= \begin{cases} i \left( \widehat w_2(i/z)\right)^{-1}, & \qquad\text{in } I \cup IV, \\
                          i \left( \widehat w_1(i/z)\right)^{-1}, & \qquad\text{in } II \cup III, \end{cases} &
w_4(z) &= \begin{cases} i \left( \widehat w_1(i/z)\right)^{-1}, & \qquad\text{in } I \cup IV, \\
                          i \left( \widehat w_2(i/z)\right)^{-1}, & \qquad\text{in } II \cup III. \end{cases}
\end{align*}
Then, the lemma follows from \cite[Lemmas 3.1--3.3]{DG} and straightforward computations.
\end{proof}


\subsection{The $\lambda$-functions}
\label{sec:lamda function}

With the functions $w_j(z)$ in Lemma \ref{lemma: w}, we define the
$\lambda$-functions as
\begin{equation}\label{def: lambda function}
\lambda_j(z)=z^2 \left( C_1 w_j(z)^4+C_2 w_j(z)^2+C_3+C_4 w_j(z)^{-2}\right),
\qquad z \in \mathcal R_j, \quad j=1,2,3,4,
\end{equation}
where
\begin{align*}
C_1&=-\tfrac32\gamma^{33/4}+\tfrac{3}{32}\gamma^{27/4}, & C_2&=\tfrac72\gamma^{21/4}+\tfrac{17}{32}\gamma^{15/4}, \\
C_3&=-\tfrac52\gamma^{9/4}-\tfrac{65}{96}\gamma^{3/4}, & C_4&=\tfrac12
\gamma^{-3/4}+\tfrac{5}{96}\gamma^{-9/4},
\end{align*}
are constants depending on $\gamma$. We now fix $\gamma$ as
\begin{equation}\label{eq:gamma explicit}
\gamma=4^{-4/3}\left(-2\mp\frac{\sigma}{|a|^{3/2}}
+\sqrt{5+4\left(1\pm\frac{\sigma}{2|a|^{3/2}}\right)^2}\right)^{4/3}>0,
\end{equation}
with $ a<0$ and $\sigma\in\mathbb{R}$. In \eqref{eq:gamma explicit}
and the rest of the paper we choose the top sign when we are in the
Brownian paths setting and the bottom sign when dealing with the
two-matrix model. It is readily seen that $\gamma$ satisfies
\begin{equation}\label{eq:ale equ for gamma}
\frac{5-16
\gamma^{3/2}}{16 \gamma^{3/4}}=1\pm\frac{\sigma}{2|a|^{3/2}}.
\end{equation}
and
\begin{equation}\label{eq:asy of gamma}
\gamma=4^{-4/3}\left(1\mp\frac{4\sigma}{9|a|^{3/2}}+\mathcal{O}(a^{-3})\right)
\end{equation}
as $a\to -\infty$ with $\sigma$ fixed.

Since $\gamma$ depends on the parameters $a$ and $\sigma$, the
functions and constants $w_j,\lambda_j,c,C_j$ as well as some
functions and constants to be introduced later, depend on $a$ and
$\sigma$ through $\gamma$, although this is not indicated
explicitly. Recall that we are interested in the behavior of these
notions for large negative $a$. We will add $^*$ to the notation to
denote these functions or constants in the limit $a\to -\infty$ with
$\sigma$ fixed.

The properties of the $\lambda$-functions are listed in the following lemma.
\begin{lemma}\label{thm: prop of lambda}
The functions $\lambda_j(z)$, $j=1,2,3,4,$ defined by \eqref{def: lambda function} have the
following properties.
\begin{itemize}
\item[\rm (a)] $\lambda_j(z)$ is analytic on $\mathcal R_j$, $j=1,2,3,4,$ and
\begin{align}
\lambda_{1,\pm}(x) &= \lambda_{3,\mp}(x), && x \in (0,\infty), \label{eq:lamda 1 lamda 3}\\
\lambda_{2,\pm}(x) &= \lambda_{4,\mp}(x), && x \in (-\infty,0), \\
\lambda_{3,\pm}(z) &= \lambda_{4,\mp}(z), && z \in (-ic,0) \cup
(0,ic). \label{eq:lamda3 4 on ic}
\end{align}
Hence the function $\bigcup_{j=1}^4 \mathcal R_j \to \C:  \mathcal
R_j \ni z \mapsto \lambda_j(z)$ has an analytic continuation to a
meromorphic function on the Riemann surface $\mathcal R $.

\item[\rm (b)] We have the following symmetry properties
\begin{align*}
\lambda_j(\overline z)&= \overline{\lambda_j(z)} & z &\in \mathcal R_j, \quad j=1,2,3,4, \\
\lambda_1(-z)&=\lambda_2(z), & z&\in \C\setminus (-\infty,0],  \\
\lambda_3(-z)&=\lambda_4(z), & z&\in \C\setminus
\left((-\infty,0]\cup[-ic, ic]\right).
\end{align*}

\item[\rm (c)]
As $z \to \infty$ within $\C\setminus [0,\infty)$, we have
\begin{align*}
\lambda_2(z) &=\tfrac23z^{3/2}+\left(1\pm\frac{\sigma}{2|a|^{3/2}}\right)z-
z^{1/2}+\ell-2Dz^{-1/2}+\mathcal O (z^{-1}), \\
\lambda_4(z) &=-\tfrac23z^{3/2}+\left(1\pm\frac{\sigma}{2|a|^{3/2}}\right)z
+z^{1/2}+\ell+2Dz^{-1/2}+\mathcal O (z^{-1}),
\end{align*}
where $\ell$ and $D$ are complex constants depending on $\gamma$.

\item[\rm (d)] For $z$ in a neighborhood of the origin we have
\begin{align*}
\lambda_2(z)& = G(z)z^{2/3}+  H(z) z^{4/3}+K(z)z^2,
\qquad\qquad~~~ \text{in } \C \setminus (-\infty,0],\\
\lambda_4(z)& =\begin{cases}
L(z),  & \qquad\text{in }I \cup IV, \\
G(z) \omega z^{2/3} +  H(z) \omega^2 z^{4/3}+K(z)z^2,  &\qquad \text{in }II, \\
G(z) \omega^2 z^{2/3}  +  H(z) \omega z^{4/3}+K(z)z^2, &\qquad
\text{in }III,
\end{cases}
\end{align*}
where $G(z)$, $H(z)$, $K(z)$, and $L(z)$ are even analytic functions in a neighborhood of the origin with
\begin{align}
G(0)&=\frac{3}{32\gamma^{5/4}}(1-16\gamma^{3/2}),     &\quad G^*(0)&=0,               \label{eq:def G0}\\
H(0)&=\frac{1}{32 \gamma^{1/4}}(25-16\gamma^{3/2}),   &\quad H^*(0)&=3\cdot 4^{-2/3}, \label{eq:def H0}\\
K(0)&=\frac{1}{96}\gamma^{3/4}(15+16\gamma^{3/2}),    &\quad K^*(0)&=\frac{1}{24},    \label{eq:def K0}\\
L(0)&=\frac{1}{96\gamma^{9/4}}(5+48\gamma^{3/2}),     &\quad L^*(0)&=\frac{16}{3}.    \label{eq:def L0}
\end{align}

\item[\rm (e)] As $z \to ic$ we have that
\begin{equation}
\lambda_{3}(z)=\wtil G(z)+\wtil H(z) (z-ic)^{3/2},\label{eq:lambda3 at ic}
\end{equation}
where we take the branch cut of $(z-ic)^{1/2}$ along $(-i\infty ,ic)$ and $-\pi/2< \arg (z-ic)<\frac{3 \pi}{2}$.
Here $\wtil G$ and $\wtil H$ are analytic functions in a neighborhood of $ic$ satisfying
\begin{align}
\wtil G(ic)&=\frac{9}{256\gamma^{9/4}}(3+16\gamma^{3/2}),       &\qquad \wtil G^*(ic^*)&=9,                              \label{eq:tileG ic} \\
\wtil H(ic)&=-\frac{3^{1/4}}{27}(5+16\gamma^{3/2})e^{3\pi i/4}, &\qquad \wtil H^*(ic^*)&=-\frac{2}{3^{7/4}}e^{3\pi i/4}. \label{eq:tileH ic}
\end{align}
\end{itemize}
\end{lemma}
Note that the complete behavior of $\lambda_j$, $j=1,2,3,4$, around the origin,
$\infty$, and $\pm ic$ can be obtained from the symmetry conditions in (b).

\begin{proof}
The proof is straightforward using \eqref{def: lambda function} and Lemma \ref{lemma: w}.
\end{proof}

As $a\to-\infty$, the functions $\lambda_j$ converge to
$\lambda_j^*$, i.e., the $\lambda$-functions associated with the
parameter $\gamma^*=4^{-4/3}$. The convergence will be uniform in
compact subsets of $\C$. For later use, we need the following
estimate of the convergence rate.

\begin{lemma}\label{lem:est of convergence}
There exists a constant $\varrho>0$ such that, for every large $a$,
we have
\[
\left| \lambda_j(z) - \lambda_j^*(z) \right| \leq \varrho
|a|^{-3/2}\max(1,|z|^{3/2}), \qquad z \in \C \setminus D(0,\delta) ,
\]
for $j=1,2,3,4$, where $D(0,\delta)$ is an open disk centered at the
origin with small radius $\delta>0$.
\end{lemma}
\begin{proof}
Due to \eqref{eq:asy of gamma}, we have
\[
\gamma-\gamma^*=\mathcal O(|a|^{-3/2}),\qquad \text{as } a\to-\infty.
\]
Together with \eqref{eq: alg eq w} and \eqref{def: lambda function}, this implies
\begin{equation}
|\lambda_j(z)-\lambda_j^*(z)|=\mathcal O(|a|^{-3/2}),\qquad \text{as }a\to-\infty,
\end{equation}
where the constant is uniform for $z$ in compact subsets of $\C$.
Combining this with the behavior of the $\lambda$-functions at
infinity, as given in Lemma \ref{thm: prop of lambda} (c), we obtain
the lemma.
\end{proof}

\section{Steepest descent analysis for $M$} \label{sec: steepest descent}

In this section we perform the steepest descent analysis of the
tacnode RH problem~\ref{rhp: tacnode rhp}. As discussed in
Sections~\ref{sec: double scaling limit 1} and \ref{sec: double
scaling limit 2}, we want to study the solution $M$ of this RH
problem for the following values of parameters
$$
r_1=r_2=1,
$$
and
\begin{equation}\label{t:scaling}
s_1=s_2=-\frac{a^2}{2},\qquad t=|a|\left(1 \pm
\frac{\sigma}{2|a|^{3/2}} \right),
\end{equation} as $a$ tends to $-\infty$. For
the Brownian paths setting it is more convenient to choose the plus
sign in the formula for $t$, whereas the minus sign is more
appropriate in the two-matrix model. The analysis consists of a
series of explicit and invertible transformations
\[ M \mapsto M^{(1)} \mapsto  M^{(2)} \mapsto  M^{(3)} \mapsto  M^{(4)} \mapsto  M^{(5)} \]
of the RH problem. In Section~\ref{sec:proof of thm} we will use
these transformations to prove Theorems~\ref{th: tacnode2Pearcey}
and \ref{th: Pearcey}. Since we are interested in the case $a \to
-\infty$, we may assume that $a<0$.

\subsection{First transformation: $M\mapsto M^{(1)}$}
This transformation is a rescaling of the RH problem for $M$. Define
\begin{equation}\label{M to A}
M^{(1)}(z;a)=\diag(|a|^{1/2},|a|^{1/2},|a|^{-1/2},|a|^{-1/2})
M\left(a^2z;-\frac{a^2}{2},|a|\left(1\pm\frac{\sigma}{2|a|^{3/2}}\right) \right).
\end{equation}
Then $M^{(1)}$ satisfies the following RH problem.
\begin{lemma} The function $M^{(1)}$ defined in \eqref{M to A} has the following properties
\begin{itemize}
\item[\rm (1)] $M^{(1)}(z)$ is analytic for $z\in\cee\setminus\Sigma_M$,
where $\Sigma_M$ is shown in Figure~\ref{fig: contour M}.
\item[\rm (2)] $M^{(1)}$ has the same jump matrices on $\Sigma_M$ as
$M$; see Figure~\ref{fig: contour M}.
\item[\rm (3)] As $z\to\infty$ with $z\in\cee\setminus\Sigma_M$, we have
\begin{multline}\label{eq:asy of M1}
M^{(1)}(z) = \left(I+\mathcal{O}(z^{-1})\right) B(z)A \\
\times \diag\left(e^{|a|^3\left(-\wtil\psi(-z)+\wtil t z\right)},
e^{|a|^3\left(-\wtil\psi(z)-\wtil tz\right)},e^{|a|^3\left(\wtil\psi(-z)+\wtil t z\right)},e^{|a|^3\left(\wtil\psi(z)-\wtil t z\right)}\right),
\end{multline}
where $A$ and $B(z)$ are given in \eqref{eq: A}--\eqref{eq: B}, and
\begin{equation}\label{def:wtilpsi}
\wtil\psi(z)=\frac{2}{3}z^{3/2}-z^{1/2}, \qquad \wtil
t=1 \pm \frac{\sigma}{2|a|^{3/2}}.
\end{equation}
\item[\rm (4)] $M^{(1)}$ is bounded near the origin.
\end{itemize}
\end{lemma}
\begin{proof}
This is immediate from RH problem \ref{rhp: tacnode rhp} and \eqref{M to A}.
\end{proof}

\subsection{Second and third transformations: $M^{(1)}\mapsto M^{(2)}\mapsto M^{(3)}$}

It is the aim of the second and the third transformation to
eliminate the jump matrices on $\Gamma_j$, $j=2,3,7,8$, see Figure
\ref{fig: contour M}. For convenience, we fix the angles
$\varphi_1=\pi/4$ and $\varphi_2=\pi/3$. In a first transformation
$M^{(1)} \mapsto M^{(2)}$ we erase the lower right block of the jump
matrices on these rays by moving them to the shifted rays emanating
from $ic$ (for $\Gamma_2$ and $\Gamma_3$) and $-ic$ (for $\Gamma_7$
and $\Gamma_8$). Thus, we introduce the new rays
\begin{equation*}
\widetilde \Gamma_{2,3}=\Gamma_{2,3}+ic, \qquad \widetilde
\Gamma_{7,8}=\Gamma_{7,8}-ic,
\end{equation*}
and define $M^{(2)}$ as follows. For $k = 2,3,7,8$ it is given by
\begin{align}\label{eq:defM(2)}
M^{(2)}(z)=
\begin{cases} M^{(1)}(z) (I-E_{3,4}), &\qquad  k=2,8, \\
              M^{(1)}(z) (I-E_{4,3}), &\qquad  k=3,7,
\end{cases}
\end{align}
in the region bounded by $\Gamma_k$, $\wtil \Gamma_k$, and $i\R$. Here $E_{i,j}$ denotes the $4 \times 4$ elementary matrix with entry $1$
at the $(i,j)$th position and all other entries equal to zero. In the remaining regions we simply set $M^{(2)} = M^{(1)}$.

In the next transformation $M^{(2)}\mapsto M^{(3)}$ we eliminate the
remaining upper left block of the jumps on $\Gamma_j$, $j=2,3,7,8$,
by, respectively, moving it to the rays $\Gamma_j$, $j=1,4,6,9$.
This completely eliminates the jump matrices on $\Gamma_j$,
$j=2,3,7,8$. Thus, for $k=1,3,6,8$, we define
\begin{align}\label{eq:defM(3)}
M^{(3)}(z)=
\begin{cases} M^{(2)}(z) (I-E_{2,1}),  &\qquad  k=1,8, \\
              M^{(2)}(z) (I-E_{1,2}),  &\qquad  k=3,6,
\end{cases}
\end{align}
for $z$ in the sector bounded by $\Gamma_k$ and $\Gamma_{k+1}$. In the other regions we put $M^{(3)} = M^{(2)}$.

A straightforward check then yields that $M^{(3)}$ has jumps on the
contour
\begin{align}\label{def:sigma 3}
\Sigma_{M^{(3)}}:= \R \cup e^{\pi i/4}\R \cup e^{-i\pi/4}\R \cup
\widetilde \Gamma_2\cup \widetilde \Gamma_3\cup \widetilde
\Gamma_7\cup \widetilde \Gamma_8 \cup [-ic,ic],
\end{align}
as illustrated in Figure \ref{fig: contour M3}, and satisfies the
following RH problem.

\begin{lemma} \label{rhp for M3}
The function  $M^{(3)}$ as defined in \eqref{eq:defM(3)} has the properties
\begin{itemize}
\item[\rm (1)] $M^{(3)}(z)$ is analytic for $z \in \C \setminus
\Sigma_{M^{(3)}}$.
\item[\rm (2)] $M^{(3)}_+(z)=M^{(3)}_-(z)J_{M^{(3)}}(z)$, for $z
\in\Sigma_{M^{(3)}}$, where the jump matrices $J_{M^{(3)}}(z)$
are shown in Figure \ref{fig: contour M3}. Note that we reversed the
orientation of some rays.
\item[\rm (3)] As $z\to \infty$, $M^{(3)}$ has the same asymptotics as $M^{(1)}$;
see \eqref{eq:asy of M1}.
\item[\rm (4)] $M^{(3)}$ is bounded near the origin and near $\pm ic$.
\end{itemize}
\end{lemma}

\begin{figure}[t]
\centering
\begin{tikzpicture}[scale=0.7]
\begin{scope}[decoration={markings,mark= at position 0.5 with {\arrow{stealth}}}]
\draw[postaction={decorate}]      (0,0)--(4,0) node[right]{$\left(\begin{smallmatrix} 0&0&1&0\\0&1&0&0\\-1&0&0&0\\0&0&0&1 \end{smallmatrix}\right)$};
\draw[postaction={decorate}]      (0,0)--(3,3) node[right]{$\left(\begin{smallmatrix} 1&0&0&0\\-1&1&0&0\\1&0&1&0\\0&0&0&1 \end{smallmatrix}\right)$};
\draw[postaction={decorate}]      (0,0)--node[very near end,right]{$\left(\begin{smallmatrix} 1&0&0&0\\ 0&1&0&0\\0&0&0&1\\0&0&-1&1 \end{smallmatrix}\right)$}(0,4) ;
\draw[postaction={decorate}]      (0,4)--(2,7.5) node[below right]{$\left(\begin{smallmatrix} 1&0&0&0\\ 0&1&0&0\\0&0&1&1\\0&0&0&1 \end{smallmatrix}\right)$};
\draw[postaction={decorate}]      (0,4)--(-2,7.5) node[below left]{$\left(\begin{smallmatrix} 1&0&0&0\\0&1&0&0\\0&0&1&0\\0&0&-1&1 \end{smallmatrix}\right)$};
\draw[postaction={decorate}]      (0,0)--(3,-3)node[right]{$\left(\begin{smallmatrix} 1&0&0&0\\1&1&0&0\\1&0&1&0\\0&0&0&1 \end{smallmatrix}\right)$};
\end{scope}
\begin{scope}[decoration={markings,mark= at position 0.5 with {\arrowreversed{stealth}}}]
\draw[postaction={decorate}]      (0,0)--(-4,0) node[left]{$\left(\begin{smallmatrix} 1&0&0&0\\0&0&0&1\\0&0&1&0\\0&-1&0&0 \end{smallmatrix}\right)$};
\draw[postaction={decorate}]      (0,0)--(-3,3) node[left]{$\left(\begin{smallmatrix} 1&-1&0&0\\0&1&0&0\\0&0&1&0\\0&1&0&1 \end{smallmatrix}\right)$};

\draw[postaction={decorate}]      (0,0)--node[very near end,right]{$\left(\begin{smallmatrix} 1&0&0&0\\ 0&1&0&0\\0&0&0&1\\0&0&-1&1 \end{smallmatrix}\right)$}(0,-4);
\draw[postaction={decorate}]      (0,-4)--(2,-7.5) node[above right]{$\left(\begin{smallmatrix} 1&0&0&0\\ 0&1&0&0\\0&0&1&1\\0&0&0&1 \end{smallmatrix}\right)$};
\draw[postaction={decorate}]      (0,-4)--(-2,-7.5) node[above left]{$\left(\begin{smallmatrix} 1&0&0&0\\0&1&0&0\\0&0&1&0\\0&0&-1&1 \end{smallmatrix}\right)$};
\draw[postaction={decorate}]      (0,0)--(-3,-3)node[left]{$\left(\begin{smallmatrix} 1&1&0&0\\0&1&0&0\\0&0&1&0\\0&1&0&1 \end{smallmatrix}\right)$};
\end{scope}
\filldraw (0,4) circle (1pt) node[left]{$ic$} (0,-4) circle (1pt) node[left]{$-ic$};
\filldraw    (0,0) circle (2pt);
\draw  (0,0) node[below right]{$0$};
\end{tikzpicture}
\caption{The jump contour $\Sigma_{M^{(3)}}$ and jump matrices for
$M^{(3)}(z)$.} \label{fig: contour M3}
\end{figure}
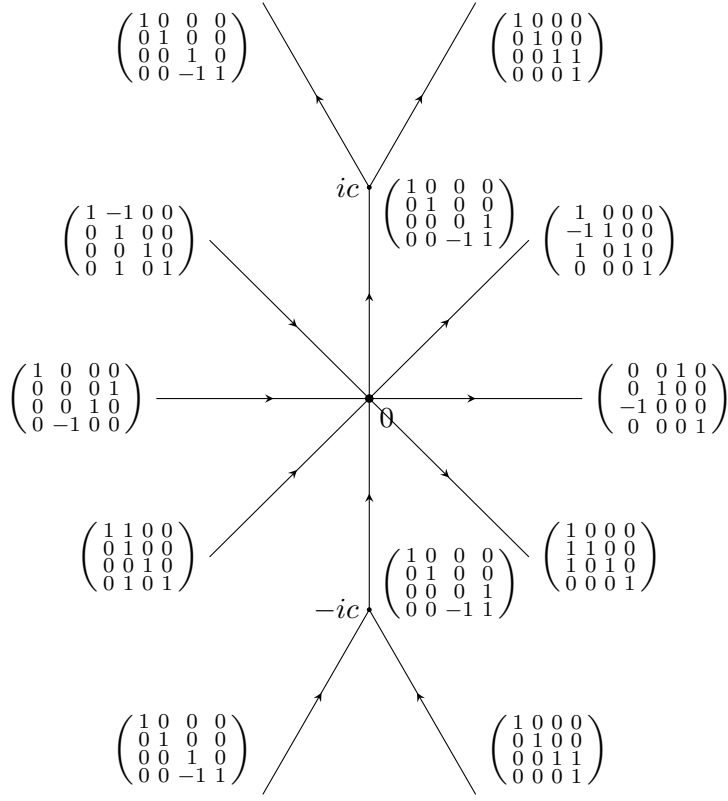

\subsection{Fourth transformation: $M^{(3)}\mapsto M^{(4)}$}
\label{sec: M3 to M4}
In the fourth transformation we normalize the
behavior around infinity using the $\lambda$-functions introduced in
Section \ref{sec:lamda function}. We define
\begin{multline}\label{eq:defM(4)}
M^{(4)}(z)=e^{-|a|^3 \ell}
\left(I-2iD|a|^3E_{3,1}+2iD|a|^3E_{4,2}\right)
M^{(3)}(z)
\\
\times \diag \left( e^{|a|^3 \lambda_1(z)},e^{|a|^3
\lambda_2(z)},e^{|a|^3 \lambda_3(z)},e^{|a|^3
\lambda_4(z)}\right),
\end{multline}
where $\ell$ and $D$ are the constants in property (c) of Lemma \ref{thm: prop of lambda}.

Then $M^{(4)}$ satisfies the following RH problem with $\Sigma_{M^{(4)}}:=\Sigma_{M^{(3)}}$.

\begin{lemma}\label{lem:rhp for M4}
The function  $M^{(4)}$, as defined in \eqref{eq:defM(4)}, has the
following properties
\begin{itemize}
\item[\rm (1)] $M^{(4)}(z)$ is analytic for  $z\in \C \setminus
\Sigma_{M^{(4)}}$.
\item[\rm (2)] For $z \in \Sigma_{M^{(4)}}$, we have
$M^{(4)}_+(z)=M^{(4)}_-(z)J_{M^{(4)}}(z)$, where
\begin{multline}\label{def:JM4}
J_{M^{(4)}}(z)= \diag \left( e^{-|a|^3
\lambda_{1,-}(z)},e^{-|a|^3 \lambda_{2,-}(z)},
e^{-|a|^3 \lambda_{3,-}(z)},e^{-|a|^3 \lambda_{4,-}(z)}\right)\\
\times J_{M^{(3)}}(z)\diag \left( e^{|a|^3
\lambda_{1,+}(z)},e^{|a|^3 \lambda_{2,+}(z)},e^{|a|^3
\lambda_{3,+}(z)},e^{|a|^3 \lambda_{4,+}(z)}\right).
\end{multline}
\item[\rm (3)] As $z \to \infty$ we have
\begin{align}\label{eq:asy of M4}
M^{(4)}(z)=\left(I+\mathcal O \left(z^{-1}\right) \right) B(z)A.
\end{align}
\item[\rm (4)]$M^{(4)}$ is bounded near the origin and near $\pm ic$.
\end{itemize}
The jump matrices $J_{M^{(4)}}(z)$ are explicitly given by
\begin{align*}
J_{M^{(4)}}(z)=\begin{pmatrix}
1 & 0 & 0 & 0\\
0 & 0 & 0 & 1\\
0 & 0 & 1 & 0\\
0 & -1 & 0 & 0
\end{pmatrix}, & \quad z \in (-\infty,0), &
J_{M^{(4)}}(z)=\begin{pmatrix}
0 & 0 & 1 & 0\\
0 & 1 & 0 & 0\\
-1 & 0 & 0 & 0\\
0 & 0 & 0 & 1
\end{pmatrix},  \quad z\in (0,\infty),
\end{align*}
\begin{align*}
J_{M^{(4)}}(z)&=
\begin{pmatrix}
1 & 0 & 0 & 0\\
0 & 1 & 0 & 0\\
0 & 0 & 0 & 1\\
0 & 0 & -1 & e^{|a|^3(\lambda_{4,+}(z)-\lambda_{4,-}(z))}
\end{pmatrix},
&  \quad z &\in (-ic,0)\cup (0,ic),
\end{align*}
and
\begin{align*}
J_{M^{(4)}}(z)&=I-e^{|a|^3(\lambda_1(z)-\lambda_2(z))}E_{2,1}
+e^{|a|^3(\lambda_1(z)-\lambda_3(z))}E_{3,1},
 &\quad z&\in\Gamma_1,
\\
J_{M^{(4)}}(z)&=I-e^{|a|^3(\lambda_2(z)-\lambda_1(z))}E_{1,2}
+e^{|a|^3(\lambda_2(z)-\lambda_4(z))}E_{4,2},
&\quad  z&\in\Gamma_4,
\\
J_{M^{(4)}}(z)&=I+e^{|a|^3(\lambda_2(z)-\lambda_1(z))}E_{1,2}
+e^{|a|^3(\lambda_2(z)-\lambda_4(z))}E_{4,2},
&\quad  z&\in\Gamma_6,
\\
J_{M^{(4)}}(z)&=I+e^{|a|^3(\lambda_1(z)-\lambda_2(z))}E_{2,1}
+e^{|a|^3(\lambda_1(z)-\lambda_3(z))}E_{3,1},
&\quad z&\in\Gamma_9,
\\
J_{M^{(4)}}(z)&=I+e^{|a|^3(\lambda_4(z)-\lambda_3(z))}E_{3,4},
&  \quad z&\in \wtil \Gamma_2\cup\wtil \Gamma_8,
\\
J_{M^{(4)}}(z)&=I-e^{|a|^3(\lambda_3(z)-\lambda_4(z))}E_{4,3}, &
\quad z&\in \wtil \Gamma_3\cup\wtil \Gamma_7.
\end{align*}
\end{lemma}
\begin{proof}
The explicit formulas of $J_{M^{(4)}}$ follow from \eqref{def:JM4} and
condition (a) in Lemma \ref{thm: prop of lambda}.

To establish the large $z$ behavior of $M^{(4)}$ shown in item
(3), we first observe from condition (c) in Lemma \ref{thm: prop of lambda} and \eqref{def:wtilpsi} that
\begin{align*}
e^{|a|^3\left(\lambda_1(z)-\wtil\psi(-z)+\tilde tz-\ell\right)}&
=1-2D|a|^3(-z)^{-1/2}+\mathcal O (z^{-1}),
\\
e^{|a|^3\left(\lambda_2(z)-\wtil\psi(z)-\tilde tz-\ell\right)}&
=1-2D|a|^3z^{-1/2}+\mathcal O (z^{-1}),
\\
e^{|a|^3\left(\lambda_3(z)+\wtil\psi(-z)+\tilde tz-\ell\right)}&
=1+2D|a|^3(-z)^{-1/2}+\mathcal O (z^{-1}),
\\
e^{|a|^3\left(\lambda_4(z)+\wtil\psi(z)-\tilde tz-\ell\right)}&
=1+2D|a|^3z^{-1/2}+\mathcal O (z^{-1}),
\end{align*}
as $z \to \infty$. Hence, from the asymptotic behavior of $M^{(3)}$
stated in item $(3)$ of Lemma~\ref{rhp for M3}, it is readily seen
that
\begin{multline} \label{M3 diag lamda }
   e^{-|a|^3 \ell}M^{(4)}(z)\diag \left( e^{|a|^3\lambda_1(z)},e^{|a|^3 \lambda_2(z)},e^{|a|^3\lambda_3(z)},e^{|a|^3 \lambda_4(z)}\right)
     =\left(I+\mathcal O (z^{-1})\right)B(z)A  \\
       \times   \left( I +   2D|a|^3\diag\left(-(-z)^{-1/2},-z^{-1/2},(-z)^{-1/2},   z^{-1/2} \right) +\diag(\mathcal O(z^{-1})) \right),
\end{multline}
as $z\to\infty$, where we recall $A$ and $B(z)$ are defined in \eqref{eq: A}--\eqref{eq: B}. Note that
\begin{multline*}
   B(z)A \left( I + 2D|a|^3\diag\left(-(-z)^{-1/2},-z^{-1/2},(-z)^{-1/2},   z^{-1/2} \right) +\diag(\mathcal O(z^{-1})) \right) \\
       =\left(I+2iD|a|^3E_{3,1}-2iD|a|^3E_{4,2} + \mathcal O (z^{-1})\right) B(z)A.
\end{multline*}
This, together with \eqref{eq:defM(4)} and \eqref{M3 diag lamda },
implies \eqref{eq:asy of M4}.
\end{proof}

\subsection{Estimate of $J_{M^{(4)}}$ on $\Sigma_{M^{(4)}}$}
Let us have a closer look at the jump matrices $J_{M^{(4)}}$ defined
in \eqref{def:JM4}. The jump matrix is constant on $(-\infty,0)$ and
on $(0,\infty)$. On $(-ic,0)\cup(0,ic)$ there is a nonconstant
$(4,4)$-entry. However, this entry is exponentially small as $a \to
-\infty$. Moreover, the decay is uniform for $z$ bounded away from
the branch points $\pm ic^*$, where $c^*=\lim_{a \to
-\infty}c=3\sqrt{3}$; see \eqref{eq: c} and \eqref{eq:asy of gamma}.
Also on the other parts of $\Sigma_{M^{(4)}}$ the nonzero
off-diagonal entries of the jump matrices turn out to be
exponentially small for large $|a|$. Again the decay is uniform if
we exclude small disks around $0$ and $\pm ic^*$.

The preceding statements can be seen as a corollary of the following
estimates. We denote with $D(z_0,\delta)$ the fixed open disk
centered at $z_0$ with small radius $\delta>0$ and let $\partial
D(z_0,\delta)$ stand for its boundary. Recall also that we fixed the
angles $\varphi_1=\pi/4$ and $\varphi_2=\pi/3$.

\begin{lemma}[Estimates for $\lambda_j$ on $\Sigma_{M^{(4)}}$]
\label{lem: estimates on Sigma 4} \textrm{}
\begin{enumerate}
\item[\rm (a)] There exist constants $c_1,c_2>0$ such that
\begin{align}
    \Re \left(\lambda_1(z)-\lambda_2(z) \right) &\leq - c_1 |z|^{2/3},&& \text{for } z\in(\Gamma_1 \cup \Gamma_9 )\setminus D(0,\delta),\\
    \Re \left(\lambda_1(z)-\lambda_3(z) \right) &\leq - c_2 |z|^{2/3},&& \text{for } z\in(\Gamma_1 \cup \Gamma_9 )\setminus D(0,\delta),\\
    \Re \left(\lambda_2(z)-\lambda_1(z) \right) &\leq - c_1 |z|^{2/3},&& \text{for } z\in(\Gamma_4 \cup \Gamma_6 )\setminus D(0,\delta),\\
    \Re \left(\lambda_2(z)-\lambda_4(z) \right) &\leq - c_2 |z|^{2/3},&& \text{for } z\in(\Gamma_4 \cup \Gamma_6 )\setminus D(0,\delta),
\end{align}
for $|a|$ large enough.
\item[\rm (b)]
There exists a constant $c_3>0$ such that
\begin{align*}
    \Re \left(\lambda_4(z)-\lambda_3(z) \right) &\leq - c_3 |z|^{3/2},&& \text{for } z\in\wtil \Gamma_2 \cup \wtil \Gamma_8 \setminus D(ic^*,\delta)\\
    \Re \left(\lambda_3(z)-\lambda_4(z) \right) &\leq - c_3 |z|^{3/2},&& \text{for } z\in\wtil \Gamma_3 \cup \wtil \Gamma_7 \setminus D(ic^*,\delta),
 \end{align*}
if $|a|$ is large enough.
\item[\rm (c)]There exists a constant $c_4>0$ such that
\[   \Re \left(\lambda_{4,+}(z)-\lambda_{4,-}(z) \right) \leq - c_4,  \]
for $z \in [i(-c^*+\delta),0)\cup(0,i(c^*-\delta)]$ if $|a|$ is
sufficiently large.
\end{enumerate}
\end{lemma}

\begin{proof}
Let us start with claim (a). It is sufficient to prove the first two
estimates on $\Gamma_1$, since the remaining estimates follow from
these ones by the symmetry conditions (b) in Lemma \ref{thm: prop of
lambda}. Moreover, due to Lemma \ref{lem:est of convergence} and the
triangle inequality it suffices to prove these estimates for the
critical $\lambda$-functions $\lambda_j^*$, $j=1,2,3,4$.

\begin{figure}[t]
\centering
\begin{overpic}[scale=0.3]{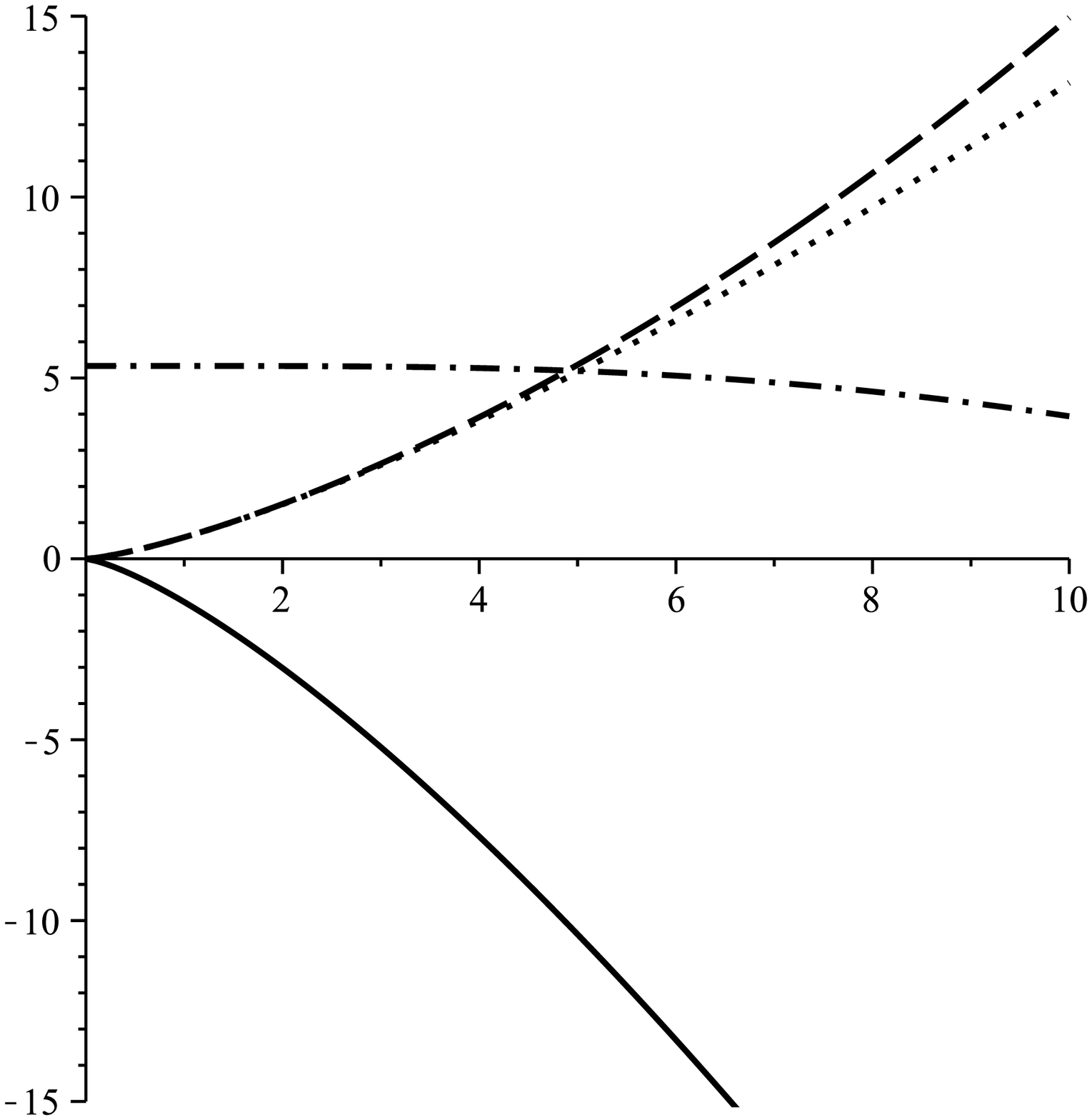}\small{\put(50,-6){(a)}}
\end{overpic}\hspace{8mm}
\begin{overpic}
[scale=0.3]{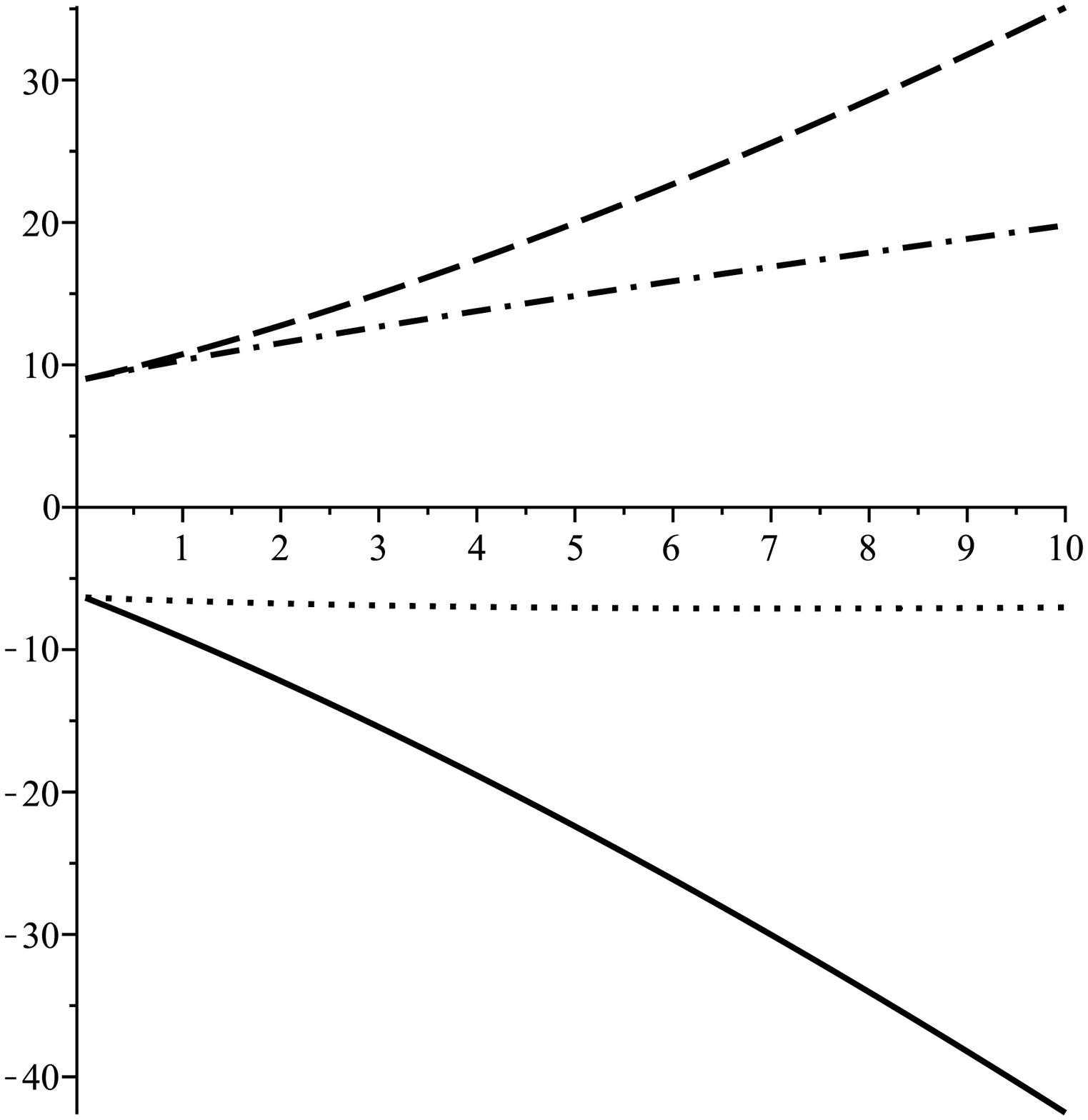}\small{\put(50,-6){(b)}}
\end{overpic}
\caption{Plots of $\Re \lambda_j^*(z)$ for $j=1$ (solid), $j=2$
(dotted), $j=3$ (dashed), and $j=4$ (dashdotted). In picture (a), $z
\in \Gamma_1$, i.e. $z=e^{\pi i/4}x$ with $x$ on the horizontal
axis. In picture (b), $z \in \wtil \Gamma_2$, i.e. $z=ic^*+e^{\pi
i/3}x$ with $x$ on the horizontal axis.} \label{fig: estimates 1}
\end{figure}

Note that $\Gamma_1=[0,\exp{i\pi/4}\infty)$. For large values of $z
\in \Gamma_1$, the estimates follow from the asymptotics of
$\lambda^*_j$, which can be obtained by taking $a\to -\infty$ in
Lemma \ref{thm: prop of lambda}(c). Using an extra argument one can
extend the estimates to $\Gamma_1 \setminus D(0,\delta)$. This is
illustrated in Figure \ref{fig: estimates 1}(a). This plots shows
the values of $\Re \lambda_j^*(z)$, $j=1,2,3,4$, for $z \in
\Gamma_1$. It is clearly seen that $\Re
(\lambda_1^*(z)-\lambda_j^*(z))<0$, $j=2,3$, for $z \in \Gamma_1
\setminus D(0,\delta)$. This completes the proof of (a).

The proof of (b) is analogous and is based on the plot in Figure
\ref{fig: estimates 1}(b). We recall that $\wtil
\Gamma_2=ic^*+[0,\exp{i\pi/3}\infty)$. Also the proof of (c) is
analogous and supported by Figure \ref{fig: estimates 2}.
\end{proof}

\begin{figure}[t]
\centering
\includegraphics[scale=0.3]{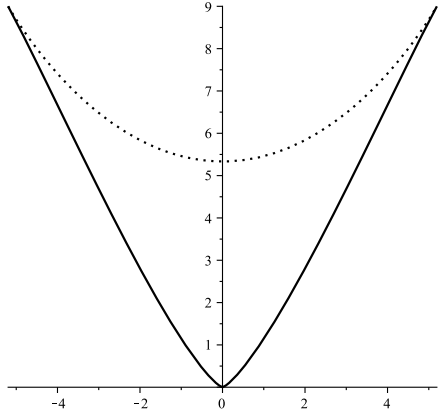}
\caption{Plots of $\lambda_{4,+}^*(z)$ (solid) and
$\lambda_{4,-}^*(z)$ (dashed) for $z \in (-ic^*,ic^*)$. Here $z=ix$
with $x$ on the horizontal axis.} \label{fig: estimates 2}
\end{figure}

Lemma \ref{lem: estimates on Sigma 4} has the following immediate corollary.

\begin{corollary} [Estimates of $J_{M^{(4)}}$ on $\Sigma_{M^{(4)}}$]\label{cor:JM4estimates}
\textrm{}
    \begin{itemize}
    \item[\rm (a)]
    There is a constant $c_1 > 0$ such that
    \[ J_{M^{(4)}}(z) = I + O\left(e^{-c_1 |a|^3 |z|^{2/3}}\right) \qquad \text{ as } a \to -\infty, \]
     uniformly for $z \in
\Sigma_{M^{(4)}}\setminus\left(\R \cup [-ic^*,ic^*]\cup D(0,\delta)
\cup D(\pm ic^*,\delta) \right)$.
    \item[\rm (b)]
    There is a constant $c_2 > 0$ such that
    \[ J_{M^{(4)}}(z) = I + O\left(e^{-c_2 |a|^3 }\right) \qquad \text{ as } a \to -\infty, \]
    uniformly for  $z \in [i(-c^*+\delta),i(c^*-\delta)]\setminus \{0\}$.
\end{itemize}
\end{corollary}

\subsection{Construction of the global parametrix $M^{(\infty)}$}

If we suppress all entries of the jump matrices for $M^{(4)}$ that
exponentially decay as $a \to -\infty$, we are led to the following
RH problem for the global parametrix $M^{(\infty)}$.

\begin{rhp}[Global parametrix]\label{rhp:global para}
We look for a $4 \times 4$ matrix valued function $M^{(\infty)}$ that satisfies
\begin{itemize}
\item[\rm (1)]
$M^{(\infty)}$ is analytic in $\C \setminus \Sigma_{M^{(\infty)}}$,
where the contour $\Sigma_{M^{(\infty)}}$ consists of the real line
oriented from left to right and the purely imaginary interval $[-ic,ic]$ oriented from bottom to top.
\item[\rm (2)] For $z\in\Sigma_{M^{(\infty)}}$, we have
\begin{align}
M^{(\infty)}_+(x) &= M^{(\infty)}_-(x) \begin{pmatrix}
0 & 0 & 1 & 0\\
0 & 1 & 0 & 0\\
-1 & 0 & 0 & 0\\
0 & 0 & 0 & 1
\end{pmatrix}, && \text{for $x \in (0,\infty)$,}  \label{eq: jump global parametrix 1}\\
M^{(\infty)}_+(x) &= M^{(\infty)}_-(x)\begin{pmatrix}
1 & 0 & 0 & 0\\
0 & 0 & 0 & 1\\
0 & 0 & 1 & 0\\
0 & -1 & 0 & 0
\end{pmatrix}, && \text{for $x \in (-\infty,0)$,} \label{eq: jump global parametrix 2}\\
M^{(\infty)}_+(z) &= M^{(\infty)}_-(z)
\begin{pmatrix}
1 & 0 & 0 & 0\\
0 & 1 & 0 & 0\\
0 & 0 & 0 & 1\\
0 & 0 & -1 & 0
\end{pmatrix},
&& \text{for $z \in (-ic,0)\cup (0,ic)$.} \label{eq: jump global parametrix 3}
\end{align}
\item[\rm (3)] As $z \to \infty$, the following asymptotic formula holds
\begin{equation}\label{eq:asy of M infty}
M^{(\infty)}(z)=\left(I+\mathcal O \left(z^{-1}\right) \right) B(z)A.
\end{equation}
\item[\rm (4)] We have
\begin{equation}\label{eq:asy of Minfty 0}
M^{(\infty)}(z)=\mathcal {O} \left(z^{-1/3}\right), \qquad \text{as }z\to 0,
\end{equation}
and
\begin{equation}\label{eq:asy of Minfty ic}
M^{(\infty)}(z)=\mathcal {O} \left((z\mp ic)^{-1/4}\right), \qquad \text{as }z\to \pm ic.
\end{equation}
\end{itemize}
\end{rhp}

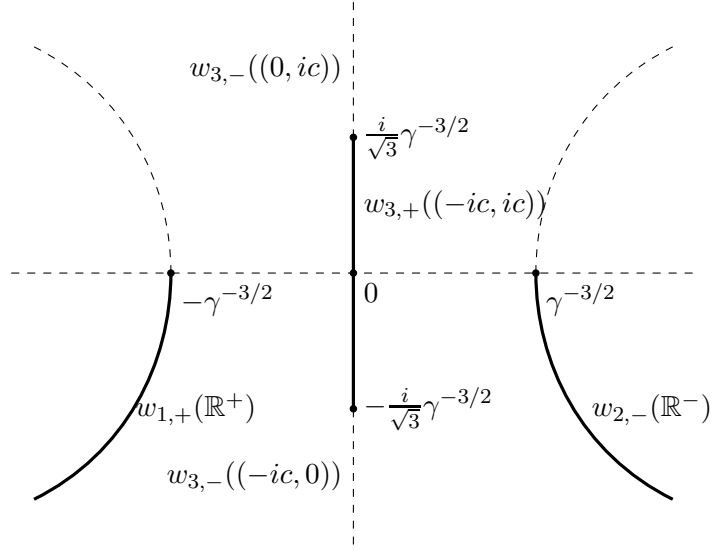
\begin{figure}[t]
\centering
\begin{tikzpicture}[scale=0.6]
\draw[dashed] (-7.5,0)--(7.5,0) ;
\draw[dashed] (-4,0).. controls (-4,2) and  (-5,4)..(-7,5)
      (4,0).. controls (4,2) and  (5,4)..(7,5);
\draw[very thick] (-4,0).. controls (-4,-2) and  (-5,-4)..(-7,-5)
           (4,0).. controls (4,-2) and  (5,-4)..(7,-5)
           (0,-3)--(0,3);
\draw[dashed]    (0,-6)--(0,-3)
                  (0,3)--(0,6);
\filldraw    (0,0) circle (2pt) (0,3) circle (2pt) (0,-3) circle (2pt) (-4,0) circle (2pt) (4,0) circle (2pt);
\draw  (0,0) node[below right]{$0$}
      (0,3) node[right]{$\frac{i}{\sqrt{3}}\gamma^{-3/2}$}
      (0,-3) node[right]{$-\frac{i}{\sqrt{3}}\gamma^{-3/2}$}
      (-4,0) node[below right]{$-\gamma^{-3/2}$}
      (4,0) node[below right]{$\gamma^{-3/2}$};
\draw (-5,-3)node[right]{$w_{1,+}(\mathbb{R}^+)$}
      (5,-3)node[right]{$w_{2,-}(\mathbb{R}^-)$}
      (0,1.5)node[right]{$w_{3,+}((-ic,ic))$}
      (0,-4.5)node[left]{$w_{3,-}((-ic,0))$}
      (0,4.5)node[left]{$w_{3,-}((0,ic))$};
\end{tikzpicture}
\caption{Image of the map $w:\mathcal R\mapsto
\overline{\mathbb{C}}$. The solid and dashed lines are the images of
the cuts in the Riemann surface $\mathcal {R}$ under this map. The
solid lines represent the branch cut of the function $\kappa(w)$
defined in \eqref{def:kappa w}.} \label{fig: image of w}
\end{figure}

An explicit solution of this RH problem can be built out of the
$w$-functions introduced in Lemma \ref{lemma: w}.

\begin{lemma}
With the functions $w_j$, $j=1,2,3,4$, defined in Lemma \ref{lemma: w},
the global parametrix $M^{(\infty)}$ is explicitly given by
\begin{equation}\label{eq: global para}
M^{(\infty)}(z) =\begin{pmatrix}
F_1(w_1(z)) & F_1(w_2(z)) & F_1(w_3(z)) & F_1(w_4(z))\\
F_2(w_1(z)) & F_2(w_2(z)) & F_2(w_3(z)) & F_2(w_4(z))\\
F_3(w_1(z)) & F_3(w_2(z)) & F_3(w_3(z)) & F_3(w_4(z))\\
F_4(w_1(z)) & F_4(w_2(z)) & F_4(w_3(z)) & F_4(w_4(z))
\end{pmatrix},
\end{equation}
where
\begin{align*}
F_1(w)&=-\frac{\gamma^{15/8}}{\sqrt{6}\kappa(w)}(w+\gamma^{-3/2})(w-\gamma^{-3/2})^2, & F_3(w)&=-\frac{i\gamma^{-3/8}}{2\sqrt{6}\kappa(w)}(w-\gamma^{-3/2})^2,\\
F_2(w)&=\frac{\gamma^{15/8}}{\sqrt{6}\kappa(w)}(w+\gamma^{-3/2})^2(w-\gamma^{-3/2}),  &
F_4(w)&=\frac{i\gamma^{-3/8}}{2\sqrt{6}\kappa(w)}(w+\gamma^{-3/2})^2,
\end{align*}
and
\begin{equation}\label{def:kappa w}
\kappa(w)=\left((w^2-\gamma^{-3})(w^2+\frac{1}{3}\gamma^{-3})\right)^{1/2}
\end{equation}
is defined in the $w$-plane with branch cut along
$w_{1,+}(\mathbb{R}^+)\cup w_{3,+}((-ic,ic))\cup w_{2,+}(\mathbb{R}^-)$ (see Figure \ref{fig: image of w}) such that
\[   \kappa(w)=w^2(1+\mathcal{O}(1/w)),    \]
as $w\to +\infty$.
\end{lemma}
\begin{proof}
The jump condition follows from item (a) in Lemma \ref{lemma: w} and \eqref{eq: global para}. A straightforward calculation with the aid of item (c), (d), and (e) in Lemma \ref{lemma: w} leads to the asymptotic behavior of
$M^{(\infty)}(z)$ in \eqref{eq:asy of M infty}, \eqref{eq:asy of Minfty 0} and \eqref{eq:asy of Minfty ic}.
\end{proof}


The global parametrix $M^{(\infty)}$ is a good approximation for
$M^{(4)}$ for values of $z$ bounded away from $0$ and $\pm i c^*$.
We will construct local parametrices near these points in the next
two sections.

\subsection{Construction of the local parametrices around $\pm ic^*$}

In the disks $D(\pm ic^*,\delta)$ around $\pm ic^*$ we build local
parametrices $M^{(\pm ic)}$ respectively. The idea is to construct
$M^{(\pm ic)}$ such that it makes exactly the same jumps in the disk
$D(\pm ic^*,\delta)$ as $M^{(4)}$, and satisfies
\begin{equation} \label{eq: matching around ic}
M^{(\pm ic)}(z)=\left(I+\mathcal O(|a|^{-3})\right)M^{(\infty)}(z),
\quad \text{uniformly for $z\in\partial D(\mp ic^*,\delta)\setminus
\Sigma_{M^{(4)}}$,}
\end{equation}
as $a \to -\infty$. This construction can be done using Airy
functions and their derivatives. As this construction is standard
(cf. \cite{Dei,DKMVZ1}) and the explicit formulas are irrelevant to
the proofs of our main theorems, we will omit it here.

\subsection{Construction of the local parametrix around the origin}

In this section we construct a local parametrix $M^{(0)}$ in the
disk $D(0,\delta)$ based on the Pearcey parametrix; see RH problem
\ref{rhp: Pearcey} below. First, we observe that the $(4,4)$-entry
of the jump for $M^{(4)}$ on the imaginary axis is exponentially
small, uniformly in a neighborhood of the origin; see item (b) of
Corollary \ref{cor:JM4estimates}. Therefore, we may ignore this
entry in the construction of the local parametrix around the origin
which then has to satisfy the following conditions.

\begin{rhp}[Local parametrix around the origin]\label{rhp:local origin}
We look for a $4 \times 4$ matrix valued function $M^{(0)}:D(0,\delta) \setminus \Sigma_{M^{(4)}} \to \C$ satisfying
\begin{itemize}
\item[\rm (1)] $M^{(0)}$ is analytic on $D(0,\delta) \setminus \Sigma_{M^{(4)}}$, where $\Sigma_{M^{(4)}}=\Sigma_{M^{(3)}}$ is defined in \eqref{def:sigma 3}.
\item[\rm(2)] For $z \in D(0,\delta) \cap \Sigma_{M^{(4)}}$, we have
\[   M^{(0)}_+(z)=M^{(0)}_-(z)J^{(0)}(z),  \]
with $J^{(0)}$ as indicated in Figure \ref{fig: contour local
parametrix}.
\item[\rm(3)]
As $a \to -\infty$, the asymptotic formula
\begin{equation}\label{eq: matching condition origin}
M^{(0)}(z)=\left(I+\mathcal O(|a|^{-3/2})\right)M^{(\infty)}(z)
\end{equation}
uniformly holds for $z\in \partial D(0,\delta) \setminus \Sigma_{M^{(4)}}$.
\item[\rm(4)] $M^{(0)}$ is bounded near the origin.
\end{itemize}
\end{rhp}

\begin{figure}[t]
\centering
\begin{tikzpicture}[scale=0.6]
\begin{scope}[decoration={markings,mark= at position 0.5 with {\arrow{stealth}}}]
\draw[postaction={decorate}]
(0,0)--(3,0) node[right]{$\begin{pmatrix} 0&0&1&0 \\ 0&1&0&0 \\ -1&0&0&0 \\ 0&0&0&1 \end{pmatrix}$};
\draw[postaction={decorate}]
(0,0)--(2.121,2.121) node[above right]
{$\begin{pmatrix} 1&0&0&0 \\ -e^{|a|^3(\lambda_1-\lambda_2)}&1&0&0 \\ e^{|a|^3(\lambda_1-\lambda_3)}&0&1&0 \\ 0&0&0&1 \end{pmatrix}$};
\draw[postaction={decorate}]
(0,0)--(2.121,-2.121) node[below right]
{$\begin{pmatrix} 1&0&0&0 \\ e^{|a|^3(\lambda_1-\lambda_2)}&1&0&0 \\ e^{|a|^3(\lambda_1-\lambda_3)}&0&1&0 \\ 0&0&0&1 \end{pmatrix}$};
\draw[postaction={decorate}]
(-3,0)node[left]
{$\begin{pmatrix} 1&0&0&0 \\ 0&0&0&1\\ 0 &0&1&0 \\ 0&-1&0 &0 \end{pmatrix}$}--(0,0) ;
\draw[postaction={decorate}]
(-2.121,2.121)node[above left]
{$\begin{pmatrix} 1&-e^{|a|^3(\lambda_2-\lambda_1)}&0&0 \\ 0&1&0&0\\ 0 &0&1&0 \\ 0&e^{|a|^3(\lambda_2-\lambda_4)} &0&1 \end{pmatrix}$}--(0,0) ;
\draw[postaction={decorate}]
(-2.121,-2.121)node[below left]
{$\begin{pmatrix} 1&e^{|a|^3(\lambda_2-\lambda_1)}&0&0 \\ 0&1&0&0\\ 0 &0&1&0 \\ 0&e^{|a|^3(\lambda_2-\lambda_4)} &0&1 \end{pmatrix}$}--(0,0) ;
\draw[postaction={decorate}]
(0,0)--(0,3) node[above]
{$\begin{pmatrix} 1&0&0&0\\ 0&1&0&0\\0&0&0&1\\0&0&-1&0 \end{pmatrix}$}(0,3);
\draw (0,0) circle (3);
\end{scope}
\begin{scope}[decoration={markings,mark= at position 0.5 with {\arrowreversed{stealth}}}]
\draw[postaction={decorate}]      (0,0)--(0,-3)node[below]
{$\begin{pmatrix} 1&0&0&0\\ 0&1&0&0\\0&0&0&1\\0&0&-1&0 \end{pmatrix}$}(0,-3);
\end{scope}
\draw  (0,0) node[below right]{$0$};
\filldraw    (0,0) circle (2pt);
\end{tikzpicture}
\caption{The jump contour and jump matrices $J^{(0)}$ for the local parametrix $M^{(0)}$ around zero. On the circle the matching condition \eqref{eq: matching condition origin} is imposed.}
\label{fig: contour local parametrix}
\end{figure}

\begin{figure}[t]
\centering
\begin{tikzpicture}[scale=0.6]
\begin{scope}[decoration={markings,mark= at position 0.5 with {\arrow{stealth}}}]
\draw[postaction={decorate}]
(0,0)--(3,0) node[right]{$\begin{pmatrix} 0&0&1&0 \\ 0&1&0&0 \\ -1&0&0&0 \\ 0&0&0&1 \end{pmatrix}$};
\draw[postaction={decorate}]
(0,0)--(2.121,2.121) node[above right]
{$\begin{pmatrix} 1&0&0&0 \\ -1&1&0&0 \\ 1&0&1&0 \\ 0&0&0&1 \end{pmatrix}$};
\draw[postaction={decorate}]
(0,0)--(2.121,-2.121) node[below right]
{$\begin{pmatrix} 1&0&0&0 \\ 1&1&0&0 \\ 1&0&1&0 \\ 0&0&0&1 \end{pmatrix}$};
\draw[postaction={decorate}]
(-3,0)node[left]
{$\begin{pmatrix} 1&0&0&0 \\ 0&0&1&0\\ 0 &-1&0&0 \\ 0&0&0 &1 \end{pmatrix}$}--(0,0) ;
\draw[postaction={decorate}]
(-2.121,2.121)node[above left]
{$\begin{pmatrix} 1&-1&0&0 \\ 0&1&0&0\\ 0 &1&1&0 \\ 0&0 &0&1 \end{pmatrix}$}--(0,0) ;
\draw[postaction={decorate}]
(-2.121,-2.121)node[below left]
{$\begin{pmatrix} 1&1&0&0 \\ 0&1&0&0\\ 0 &1&1&0 \\ 0&0 &0&1 \end{pmatrix}$}--(0,0) ;
\end{scope}
\draw (0,0) circle (3);
\draw  (0,0) node[below ]{$0$};
\filldraw    (0,0) circle (2pt);
\end{tikzpicture}
\caption{The jump contour and jump matrices for $\wtil M^{(0)}$ around zero.}
\label{fig: contour local parametrix til}
\end{figure}
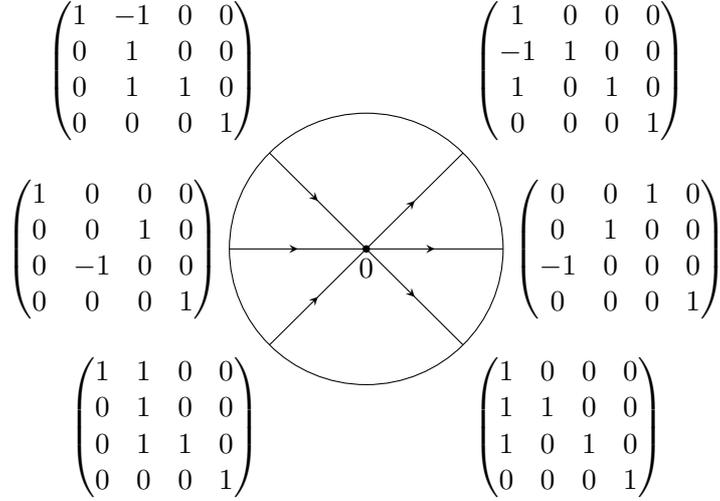

In a first step to solve the RH problem \ref{rhp:local origin}, we
make a transformation that reduces the jumps of $M^{(0)}$ to
constant matrices. Moreover we want the jumps on the imaginary axis
to disappear. To this end we define
\begin{equation}\label{def:wtil M0}
\wtil M^{(0)}(z)= M^{(0)}(z) \diag \left( e^{-|a|^3\lambda_1(z)},e^{-|a|^3\lambda_2(z)},e^{-|a|^3\lambda_3(z)},e^{-|a|^3\lambda_4(z)}\right) C_\pm^{-1},
\end{equation}
for $\pm \Re z >0$, where $C_{\pm}$ are two constant matrices given
by
\begin{equation} \label{eq: Cpm}
C_+=I_4, \qquad~~ C_-=\diag\left(1,1,\begin{pmatrix} 0 & 1 \\ -1 & 0
\end{pmatrix}\right).
\end{equation}
A straightforward check shows that the jumps for $\wtil M^{(0)}$, as
indicated in Figure \ref{fig: contour local parametrix til}, are
indeed constant. Furthermore, all jump matrices have the block form
$\begin{pmatrix}Q&0\\0&1\end{pmatrix}$, where $Q$ is a non-trivial
$3\times 3$ constant matrix. Combined with the asymptotic behavior
of the $\lambda$-functions around the origin (see items (b) and (d)
in Lemma \ref{thm: prop of lambda}), this leads us to propose the
following model RH problem.

\begin{rhp}[Model RH problem for the local parametrix around the origin] \label{rhp:modelorigin}
We look for a $3 \times 3$ matrix valued function $\Phi(\cdot;\rho)$ depending on a parameter $\rho\in\mathbb{R}$ that satisfies the following conditions.
\begin{itemize}
\item[\rm (1)] $\Phi$ is defined and analytic on $\C \setminus \Sigma_{\Phi}$,
where the contour $\Sigma_{\Phi}$ consists of $6$ semi-infinite rays
as shown in Figure \ref{fig: model rhp}.
\item[\rm (2)] For $z\in\Sigma_\Phi$, we have
\[
\Phi_+(z)= \Phi_-(z) J_{ \Phi}(z),
\]
where the jump matrix $J_{ \Phi}$ is constant on each ray and
specified in Figure \ref{fig: model rhp}.
\item[\rm (3)] As $z \to \infty$ and $\pm\Im z >0$, we have
\begin{equation}\label{eq:asy of mod RHP}
\Phi(z)=\sqrt{\frac{1}{6\pi}}i e^{-\rho^2/8} \diag
\left(z^{1/3},1,z^{-1/3} \right) L_{\pm} \left(I+\mathcal
O(z^{-2/3}) \right)e^{ \Theta(z;\rho)},
\end{equation}
where $L_{\pm}$ are constant matrices
\begin{align}\label{eq:def Lpm}
L_{+}=
\begin{pmatrix}
-\omega & -1 & -\omega^2 \\ -1&-1&-1 \\ -\omega^2 & -1& -\omega
\end{pmatrix}, \qquad
L_{-}=
\begin{pmatrix}
-\omega^2 & -1& \omega \\ -1&-1&1 \\ -\omega & -1 &  \omega^2
\end{pmatrix},
\end{align}
and $\Theta(z;\rho)$ is given by
\begin{equation*}
\Theta(z;\rho)= \begin{cases}
\diag (-\theta_2(z;\rho),-\theta_3(z;\rho),-\theta_1(z;\rho)), & \text{for $\Im z >0$,} \\
\diag (-\theta_1(z;\rho),-\theta_3(z;\rho),-\theta_2(z;\rho)), & \text{for $\Im z <0$,} \\
\end{cases}
\end{equation*}
with
\begin{equation}\label{def:theta k}
\theta_k(z;\rho)=\frac34 \omega^{2k}z^{4/3}+\frac{\rho}{2}\omega^kz^{2/3},\quad k=1,2,3.
\end{equation}
\item[\rm (4)] $\Phi(z)$ is bounded near the origin.
\end{itemize}
\end{rhp}
At this moment it is not clear where the exact formulation of the
asymptotics in item $(3)$ comes from. This follows a posteriori from
the solution of the RH problem given in Lemma \ref{lemma: solution
model rhp}.

\begin{figure}[t]
\centering
\begin{tikzpicture}[scale=0.6]
\begin{scope}[decoration={markings,mark= at position 0.5 with {\arrow{stealth}}}]
\draw[postaction={decorate}]
(0,0)--(4,0) node[right]{$\left(\begin{matrix} 0&0&1 \\ 0&1&0 \\ -1&0&0 \end{matrix}\right)$};
\draw[postaction={decorate}]
(0,0)--(3,3) node[right]{$\left(\begin{matrix} 1&0&0 \\ -1&1&0 \\ 1&0&1 \end{matrix}\right)$};
\draw[postaction={decorate}]
(0,0)--(3,-3) node[right]{$\left(\begin{matrix} 1&0&0 \\ 1&1&0 \\ 1&0&1 \end{matrix}\right)$};
\draw[postaction={decorate}]
(-4,0)node[left]{$\left(\begin{matrix} 1&0&0 \\ 0&0&1 \\ 0&-1&0 \end{matrix}\right)$}--(0,0) ;
\draw[postaction={decorate}]
(-3,3)node[left]{$\left(\begin{matrix} 1&-1&0 \\ 0&1&0 \\ 0&1&1 \end{matrix}\right)$}--(0,0) ;
\draw[postaction={decorate}]
(-3,-3)node[left]{$\left(\begin{matrix} 1&1&0 \\ 0&1&0 \\ 0&1&1 \end{matrix}\right)$}--(0,0) ;
\end{scope}
\draw  (0,0) node[below]{$0$};
\filldraw    (0,0) circle (2pt);
\draw (5mm,0mm) arc(0:45:5mm);
\draw (1.3,0) node[above]{$\pi/4$};
\end{tikzpicture}
\caption{The jump contour $\Sigma_{\Phi}$ and jump matrices for the model RH problem $ \Phi$.}
\label{fig: model rhp}
\end{figure}
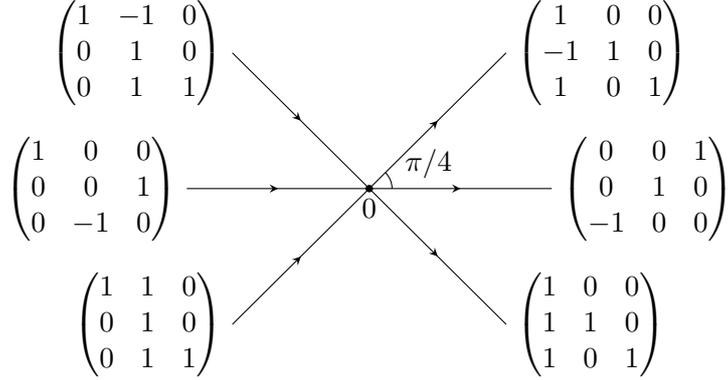

Note that $\Phi$ depends on the parameter $\rho$ through the
asymptotics \eqref{eq:asy of mod RHP}. It turns out that we can
solve this model RH problem using the Pearcey parametrix that was introduced by Bleher and Kuijlaars in the context of random matrices with external source \cite{BK3}. The Pearcey parametrix is the unique solution to the following RH problem.

\begin{rhp}[Pearcey parametrix] \label{rhp: Pearcey}
\textrm{}
\begin{itemize}
\item[\rm (1)] $\Phi^{\Pe}(\cdot;\rho)$ is a $3\times 3$ matrix valued function depending a parameter
$\rho\in\mathbb{R}$, defined and analytic on $\C \setminus
\Sigma_\Phi$.
\item[\rm (2)] For $z\in \Sigma_\Phi^{\Pe}$, the function $\Phi^{\Pe}$ has the jump
\[
\Phi^{\Pe}_+(z)=\Phi^{\Pe}_-(z) J_\Phi^{\Pe}(z),
\]
where the jump matrix $J_\Phi^{\Pe}$ is constant on each ray and specified in Figure \ref{fig: contour Pearcey}.
\item[\rm (3)]  As $z \to \infty$ and $\pm\Im z >0$, we have
\[
\Phi^{\Pe}(z)=
\sqrt{\frac{2\pi}{3}}i e^{\rho^2/8}\diag \left(z^{-1/3},1,z^{1/3} \right)L_{\pm}^{\Pe}
\left(I+\mathcal O(z^{-2/3}) \right)e^{\Theta(z;\rho)},
\]
where $L_{\pm}^{\Pe}$ are constant matrices
\begin{align}
L_{+}^{\Pe}=
\begin{pmatrix}
-\omega & \omega^2 & 1 \\ -1&1&1 \\ -\omega^2 & \omega & 1
\end{pmatrix},
\qquad
L_{-}^{\Pe}=
\begin{pmatrix}
\omega^2 & \omega & 1 \\ 1&1&1 \\ \omega & \omega^2 & 1
\end{pmatrix},
\end{align}
and $\Theta(z)$ given by
\begin{align}
\Theta(z;\rho)&= \begin{cases}
\diag (\theta_1(z;\rho),\theta_2(z;\rho),\theta_3(z;\rho)), & \text{for $\Im z >0$,} \\
\diag (\theta_2(z;\rho),\theta_1(z;\rho),\theta_3(z;\rho)), & \text{for $\Im z <0$,} \\
\end{cases}
\end{align}
and where $\theta_k(z;\rho)$ is defined in \eqref{def:theta k}.
\item[\rm (4)] $\Phi^{\Pe}(z)$ is bounded near the origin.
\end{itemize}
\end{rhp}

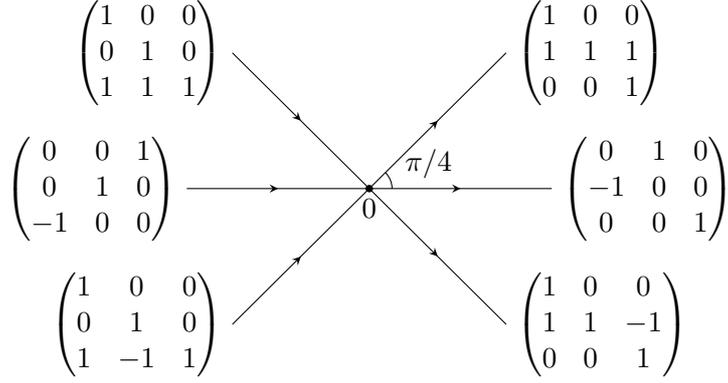
\begin{figure}[t]
\centering
\begin{tikzpicture}[scale=0.6]
\begin{scope}[decoration={markings,mark= at position 0.5 with {\arrow{stealth}}}]
\draw[postaction={decorate}]
(0,0)--(4,0) node[right]{$\left(\begin{matrix} 0&1&0 \\ -1&0&0 \\ 0&0&1 \end{matrix}\right)$};
\draw[postaction={decorate}]
(0,0)--(3,3) node[right]{$\left(\begin{matrix} 1&0&0 \\ 1&1&1 \\ 0&0&1 \end{matrix}\right)$};
\draw[postaction={decorate}]
(0,0)--(3,-3) node[right]{$\left(\begin{matrix} 1&0&0 \\ 1&1&-1 \\ 0&0&1 \end{matrix}\right)$};
\draw[postaction={decorate}]
(-4,0)node[left]{$\left(\begin{matrix} 0&0&1 \\ 0&1&0 \\ -1&0&0 \end{matrix}\right)$}--(0,0) ;
\draw[postaction={decorate}]
(-3,3)node[left]{$\left(\begin{matrix} 1&0&0 \\ 0&1&0 \\ 1&1&1 \end{matrix}\right)$}--(0,0) ;
\draw[postaction={decorate}]
(-3,-3)node[left]{$\left(\begin{matrix} 1&0&0 \\ 0&1&0 \\ 1&-1&1 \end{matrix}\right)$}--(0,0) ;
\end{scope}
\draw  (0,0) node[below]{$0$};
\filldraw    (0,0) circle (2pt);
\draw (5mm,0mm) arc(0:45:5mm);
\draw (1.3,0) node[above]{$\pi/4$};
\end{tikzpicture}
\caption{The jump contour and jump matrices for the Pearcey parametrix $\Phi^{\Pe}$.}
\label{fig: contour Pearcey}
\end{figure}
The above RH problem has a unique solution $\Phi^{\Pe}$, given in terms of solutions of the Pearcey differential equation
\begin{equation*}
y'''(z)-\rho p'(z)-zy(z)=0.
\end{equation*}
Since the exact formula of $\Phi^{\Pe}$ is not relevant for our purposes, we decide not to write it down and refer to \cite[Section
8.1]{BK3} for details.

The next lemma states that RH problem \ref{rhp:modelorigin} can be
solved in terms of the Pearcey parametrix.

\begin{lemma} \label{lemma: solution model rhp}
Let $\Phi^{\Pe}$ be the Pearcey parametrix as given in RH problem
\ref{rhp: Pearcey}, the solution of RH problem
\ref{rhp:modelorigin} is given by
\begin{equation}\label{eq:sol of modelRHP}
\Phi(z) = \Phi^{\Pe}(z)^{-T}\begin{pmatrix} 0&0&-1 \\ 1&0&0 \\ 0&1&0 \end{pmatrix},
\end{equation}
where the superscript $^{-T}$ stands for inverse transpose.
\end{lemma}
\begin{proof}
It is straightforward to check that the right-hand side of
\eqref{eq:sol of modelRHP} satisfies RH problem
\ref{rhp:modelorigin}. Then \eqref{eq:sol of modelRHP}
follows from the uniqueness of the solution to this RH problem.
\end{proof}

Now, we are ready to construct the local parametrix $M^{(0)}$. In
view of \eqref{def:wtil M0}, we look for a parametrix in the form
\begin{multline} \label{eq: local parametrix}
M^{(0)}(z)=E_0(z)
\begin{pmatrix}  \Phi(f(z);\rho(z)) & 0 \\ 0 & \phi(z)
\end{pmatrix}C_\pm \\
\times
\diag \left( e^{|a|^3\lambda_1(z)},e^{|a|^3\lambda_2(z)},e^{|a|^3\lambda_3(z)},e^{|a|^3\lambda_4(z)}\right),
\qquad \pm \Re z>0,
\end{multline}
where $E_0(z)$ is an analytic prefactor, $f(z)$ is a conformal map
near the origin, $\phi(z)$ is an analytic function and $C_\pm$ are
the constant matrices shown in \eqref{eq: Cpm}. The parameter $\rho$
will also be chosen to be dependent on $z$.

Recall the functions $G(z)$, $H(z),$ and $K(z)$ introduced in
item $(d)$ of Lemma \ref{thm: prop of lambda}. We first define
\begin{equation*}
f(z)=f(z;a)=|a|^{9/4}\left( \frac43 H(z) \right)^{3/4}z, \qquad z \in D(0,\delta).
\end{equation*}
By \eqref{eq:def H0}, it is readily seen that
\begin{equation} \label{eq: conf f}
f(z)=|a|^{9/4}\left( \frac43 H(0) \right)^{3/4}z+\mathcal O(z^3),
\end{equation}
as $z \to 0$. Hence, $f$ is a conformal map in a neighborhood of the
origin. We then modify the contours $\left(\Sigma_{M^{(4)}}\setminus
i\mathbb{R}\right)\cap D(0,\delta)$ if necessary in such a way that
$f$ maps them into $\Sigma_{\Phi}$. Next we set
\begin{equation} \label{eq: def of rho}
\rho(z)=\rho(z;a)=2|a|^3\frac{G(z)z^{2/3}}{f(z)^{2/3}},\qquad z\in
D(0,\delta),
\end{equation} and
\[
\phi(z)=\phi(z;a)=\left\{
           \begin{array}{ll}
             e^{-|a|^3\lambda_{4}(z)}, & \qquad\hbox{$\Re z>0$,} \\
             e^{-|a|^3\lambda_{3}(z)}, & \qquad \hbox{$\Re z<0$.}
           \end{array}
         \right.
\]
Also the function $\phi(z)$ is analytic on $D(0,\delta)$, which
follows from \eqref{eq:lamda3 4 on ic}. With these definitions of
$f(z)$ and $\rho(z)$, it is easily seen from items (b) and (d) in
Lemma \ref{thm: prop of lambda} that
\begin{equation}\label{eq:theta 123}
\begin{aligned}
\theta_{1}(f(z);\rho(z))&=
      \begin{cases}
        |a|^3\left(\lambda_3(z)-K(z)z^2\right), &\quad \hbox{for $z\in D(0,\delta)\cap I$,} \\
        |a|^3\left(\lambda_4(z)-K(z)z^2\right), &\quad \hbox{for $z\in D(0,\delta)\cap II$,} \\
        |a|^3\left(\lambda_1(z)-K(z)z^2\right), &\quad \hbox{for $z\in D(0,\delta)\cap (III\cup IV)$,}
      \end{cases} \\
\theta_{2}(f(z);\rho(z))&=
      \begin{cases}
        |a|^3\left(\lambda_1(z)-K(z)z^2\right), &\quad \hbox{for $z\in D(0,\delta)\cap (I\cup II)$,} \\
        |a|^3\left(\lambda_4(z)-K(z)z^2\right), &\quad \hbox{for $z\in D(0,\delta)\cap III$,} \\
        |a|^3\left(\lambda_3(z)-K(z)z^2\right), &\quad \hbox{for $z\in D(0,\delta)\cap IV$,}
      \end{cases}  \\
\theta_{3}(f(z);\rho(z))&=
|a|^3\left(\lambda_2(z)-K(z)z^2\right), \qquad\quad \hbox{for $z\in D(0,\delta)\setminus (-\delta,0]$.}
\end{aligned}
\end{equation}
Finally, the matching condition \eqref{eq: matching condition
origin} leads us to the definition of the prefactor $E_0(z)$ as
\begin{multline}\label{eq: E0}
E_0(z)=M^{(\infty)}(z)(C_\pm)^{-1} \\
\times
\begin{pmatrix} -i\sqrt{6\pi}e^{\rho(z)^2/8-|a|^3K(z)z^2} (L_{\pm})^{-1}
\diag \left( f(z)^{-1/3},1,f(z)^{1/3}\right) & 0
\\ 0 & 1
\end{pmatrix}, \quad \pm \Re z>0,
\end{multline}
where $L_{\pm}$ is defined in \eqref{eq:def Lpm}. We then have the following lemma.
\begin{lemma}
$E_0(z)$ has an analytic continuation to $D(0,\delta)$.
\end{lemma}
\begin{proof}
$E_0(z)$ is clearly analytic in $D(0,\delta) \setminus (\R \cup i\R)$.
For $z \in i\R \setminus \{0\}$ we have that $E_{0,+}(z)=E_{0,-}(z)$ as
a consequence of \eqref{eq: jump global parametrix 3} and \eqref{eq: Cpm}.
There is no jump on $\R^+$ either, because of \eqref{eq: jump global parametrix 1} and the observation
\[
 L_+ =  L_- \begin{pmatrix} 0&0&-1 \\ 0&1&0 \\ 1&0&0 \end{pmatrix}.
\]
With a bit more effort one can also check that the jump on $\R^-$ is trivial.
Hence $E_0(z)$ is analytic in the punctured disk $D(0,\delta)\setminus \{0\}$.
It follows from the behavior of the global parametrix $M^{(\infty)}$ around zero
that the singularity of $E_0(z)$ at the origin is removable.
\end{proof}

It is now straightforward to check that the jumps of
$M^{(0)}$ on $D(0,\delta) \cap \Sigma_{M^{(4)}}$ and its large $|a|$ behavior on $\partial D(0,\delta)$
(with the aid of \eqref{eq:asy of mod RHP} and \eqref{eq:theta
123}) are indeed of our required form. In summary, we have the following lemma.

\begin{lemma}
The matrix valued function $M^{(0)}$ defined in \eqref{eq: local
parametrix} satisfies conditions (1)--(4) of RH problem
\ref{rhp:local origin}.
\end{lemma}

\subsection{Final transformation: $M^{(4)}\mapsto M^{(5)}$}

Using the global parametrix $M^{(\infty)}$ and the local
parametrices $M^{(\pm ic)}$ and $M^{(0)}$, we define the fifth
transformation $M^{(4)}\mapsto M^{(5)}$ as follows
\begin{equation} \label{eq:defM(5)}
M^{(5)}(z)=\begin{cases}
M^{(4)}(z)\left( M^{(\infty)}(z) \right)^{-1} &
\text{for $z \in \C \setminus (\Sigma_{M^{(4)}}\cup D(0,\delta) \cup D(\pm ic^*,\delta))$}, \\
M^{(4)}(z)\left( M^{(0)}(z) \right)^{-1}      & \text{for $z \in D(0,\delta) \setminus \Sigma_{M^{(4)}}$}, \\
M^{(4)}(z)\left( M^{(\pm ic)}(z) \right)^{-1} & \text{for $z \in
D(\pm ic^*,\delta) \setminus \Sigma_{M^{(4)}}$}.
\end{cases}
\end{equation}

Then $M^{(5)}$ is defined and analytic outside of $\Sigma_{M^{(4)}}$
and the three disks around $0$ and $\pm ic^*$, with an analytic
continuation across those parts of $\Sigma_{M^{(4)}}$ where the
jumps of the parametrices coincide with those of $M^{(4)}$. What
remains are the jumps on a contour $\Sigma_{M^{(5)}}$ that consists
of the three circles around $0$, $\pm ic^*$, the parts of
$\Gamma_{1}$, $\wtil \Gamma_{2}$, $\wtil \Gamma_{3}$, $\Gamma_{4}$,
$\Gamma_{6}$, $\wtil \Gamma_{7}$, $\wtil \Gamma_{8}$ and
$\Gamma_{9}$ outside of the disks, and
$(i(-c^*+\delta),i(c^*-\delta))$. The circles are oriented
clockwise. Then $M^{(5)}$ satisfies the following RH problem.

\begin{rhp} \label{rhpforR} \textrm{ }
\begin{enumerate}
\item[\rm (1)] $M^{(5)}$ is defined and analytic in $\mathbb C \setminus \Sigma_{M^{(5)}}$.
\item[\rm (2)] For $z\in\Sigma_M^{(5)}$, we have
$$M^{(5)}_+(z) = M^{(5)}_-(z) J_{M^{(5)}}(z),$$
where
\begin{align*}
    J_{M^{(5)}}(z)  = \begin{cases}
            M^{(0)}(z) (M^{(\infty)}(z))^{-1} &  \textrm{for $z\in\partial D(0, \delta)$, }   \\
            M^{(\pm i c)}(z) (M^{(\infty)}(z))^{-1} & \textrm{for $z\in\partial D(\pm i c^*,\delta)$},  \\
            M^{(0)}_-(z) J_{M^{(4)}}(z) (M^{(0)}_+(z))^{-1}  &  \textrm{for $z\in(-i\delta, i\delta)$},  \\
            M^{(\infty)}_-(z) J_{M^{(4)}}(z) (M^{(\infty)}_+(z))^{-1} &\textrm{elsewhere on  $\Sigma_{M^{(5)}}$}.
            \end{cases}
\end{align*}
 \item[\rm (3)] As $z\to\infty$, we have
\begin{equation*}
        M^{(5)}(z) = I+O(1/z).
\end{equation*}
\end{enumerate}
\end{rhp}

Furthermore, the jump matrix $J_{M^{(5)}}$ tends to the identity matrix on
$\Sigma_{M^{(5)}}$ as $a \to-\infty$, both uniformly and in $L^2$-sense.
Indeed, from the matching conditions \eqref{eq: matching around ic} and
\eqref{eq: matching condition origin}, we see that
\begin{equation}
J_{M^{(5)}}(z)=I+\mathcal O (|a|^{-3/2})
\end{equation}
as $a\to -\infty$, uniformly on the circles around
$0$ and $\pm ic^*$. On the remaining parts of $\Sigma_{M^{(5)}}$, $J_{M^{(5)}}$ is uniformly
exponentially small (see Corollary \ref{cor:JM4estimates}). Then, as
in \cite{DKMVZ992,DKMVZ1}, we conclude that
\begin{equation}\label{eq:Restimate}
M^{(5)}(z)=I+O\left(\frac{1}{|a|^{3/2}(|z|+1)}\right),
\end{equation}
as $a\to -\infty$, uniformly for $z$ in the complex plane outside of $\Sigma_{M^{(5)}}$.

The estimate \eqref{eq:Restimate} is the main outcome of our
Deift-Zhou steepest descent analysis for RH problem \ref{rhp:
tacnode rhp}. We will use it to prove our main results in the next
section.

\section{Proofs of the main theorems}\label{sec:proof of thm}

\subsection{RH formula for the Pearcey kernel}
\label{sec: RH of Pearcey kernel}

Before coming to the proofs of the main theorems, we state the following lemma
giving an expression of the Pearcey kernel in terms of the Pearcey parametrix.

\begin{lemma}
Let $K^\Pe(x,y;\rho)$ be the Pearcey kernel defined in \eqref{eq:
pearcey kernel}, we have
\begin{equation} \label{eq: Pearcey kernel}
K^\Pe(x,y;\rho)=\begin{cases}
\frac{1}{2 \pi i (x-y)}  \begin{pmatrix} -1 & 1 & 0  \end{pmatrix} \Phi_+^\Pe(y;\rho)^{-1} \Phi_+^\Pe(x;\rho)  \begin{pmatrix} 1\\1\\0 \end{pmatrix}, & \text{for }x,y>0, \\
\frac{1}{2 \pi i (x-y)}  \begin{pmatrix} -1 & 0 & 1  \end{pmatrix} \Phi_+^\Pe(y;\rho)^{-1} \Phi_+^\Pe(x;\rho)  \begin{pmatrix} 1\\1\\0 \end{pmatrix}, & \text{for }x>0,y<0, \\
\frac{1}{2 \pi i (x-y)}  \begin{pmatrix} -1 & 1 & 0  \end{pmatrix} \Phi_+^\Pe(y;\rho)^{-1} \Phi_+^\Pe(x;\rho)  \begin{pmatrix} 1\\0\\1 \end{pmatrix}, & \text{for }x<0,y>0, \\
\frac{1}{2 \pi i (x-y)}  \begin{pmatrix} -1 & 0 & 1  \end{pmatrix}
\Phi_+^\Pe(y;\rho)^{-1} \Phi_+^\Pe(x;\rho)  \begin{pmatrix} 1\\0\\1
\end{pmatrix}, & \text{for }x,y<0,
\end{cases}
\end{equation}
where $\Phi^\Pe(\cdot;\rho)$ is the the unique solution of RH
problem \ref{rhp: Pearcey}.
\end{lemma}
\begin{proof}
See \cite[Section 10.2]{BK3}.
\end{proof}

\subsection{Proof of Theorem \ref{th: tacnode2Pearcey}}
\label{sec:proof of tacnode}

First we prove \eqref{eq: tacnode th 2}. We will focus on the case $u,v>0$, or equivalently,
$x,y>0$. Other cases can be proved similarly. We start from
\eqref{tacnodekernel:pos}, and the strategy is to express this
kernel in terms of $M^{(5)}$ instead of $M$ by unfolding all
transformations $M\mapsto M^{(1)}\mapsto M^{(2)}\mapsto
M^{(3)}\mapsto M^{(4)} \mapsto M^{(5)}$ of the steepest descent
analysis.

From the first transformation $M\mapsto M^{(1)}$ in \eqref{M to A},
it follows
\begin{multline*}
a^2 \Ktac \left(a^2u,a^2v;-\tfrac12{a^2},|a|\left(1+\frac{\sigma}{2|a|^{3/2}}\right)\right)\\
=\frac{1}{2 \pi i (u-v)} \begin{pmatrix} -1 & 0 & 1 & 0 \end{pmatrix} {M_+^{(1)}\left(v; a\right)}^{-1} {M^{(1)}_+\left(u;
a\right)} \begin{pmatrix} 1\\0\\1\\0 \end{pmatrix},
\end{multline*}
where, since we work in Brownian paths model, we chose the top sign
in $\pm \sigma$. The transformations $M^{(1)}\mapsto M^{(2)}\mapsto
M^{(3)}$ in \eqref{eq:defM(2)} and \eqref{eq:defM(3)} leave this
formula essentially unaffected. Applying the transformation
$M^{(3)}\mapsto M^{(4)}$ in \eqref{eq:defM(4)}, however, yields
\begin{multline*}
a^2 \Ktac \left(a^2u,a^2v;a,-\tfrac12{a^2},|a|\left(1+\frac{\sigma}{2|a|^{3/2}}\right)\right)\\
=\frac{1}{2 \pi i (u-v)} \begin{pmatrix} -e^{|a|^3\lambda_{1,+}(v)}
& 0 & e^{|a|^3\lambda_{3,+}(v)} & 0 \end{pmatrix}
M^{(4)}_+(v;a)^{-1}M_+^{(4)}(u;a) \begin{pmatrix}
e^{-|a|^3\lambda_{1,+}(u)}\\0\\e^{-|a|^3\lambda_{3,+}(u)}\\0
\end{pmatrix}.
\end{multline*}
In the next step we unfold the transformation $M^{(4)}\mapsto M^{(5)}$ in \eqref{eq:defM(5)}. Assuming that $0<u,v < \delta$, it follows from \eqref{eq: local parametrix} that
\begin{multline*}
a^2\Ktac \left(a^2u,a^2v;-\tfrac12 a^2,|a|\left(1+\frac{\sigma}{2|a|^{3/2}}\right)\right)=\frac{1}{2 \pi i (u-v)}  \\
\times \begin{pmatrix} -1 & 0 & 1 & 0 \end{pmatrix}\begin{pmatrix}
\Phi(f(v);\rho(v))^{-1} & 0 \\ 0&\phi(v)^{-1}\end{pmatrix}
E_0(v)^{-1}M^{(5)}(v)^{-1} \\ \times
M^{(5)}(u)E_0(u)\begin{pmatrix} \Phi(f(u);\rho(u)) & 0 \\ 0&\phi(u)\end{pmatrix}  \begin{pmatrix} 1\\0\\1\\0 \end{pmatrix}.
\end{multline*}

Now we fix $x,y>0$ and take
\begin{equation*}
u=2^{-1/2}|a|^{-9/4}x,\qquad  v=2^{-1/2}|a|^{-9/4}y,
\end{equation*}
so that $0<u,v <\delta$ for $|a|$ sufficiently large. Under this change of variables, it follows from \eqref{eq: conf f} and
\eqref{eq:def H0} that
\begin{align*}
f(u)\to x, \qquad f(v)\to y,
\end{align*}
as $a \to -\infty$. Also \eqref{eq: def of rho}, \eqref{eq: conf f}, \eqref{eq:def H0},
\eqref{eq:def G0}, and \eqref{eq:gamma explicit} imply
\begin{align*}
\rho(u)\to  \sigma, \qquad \rho(v)\to  \sigma,
\end{align*}
as $a \to -\infty$. Furthermore, by standard considerations it
follows that
\[
M^{(5)}(v;a)^{-1} M^{(5)}(u;a)=I+\mathcal O \left(
\frac{v-u}{|a|^{3/2}}\right)=I+\mathcal O \left(
\frac{x-y}{|a|^{15/4}}\right),
\]
as $ a \to -\infty$, uniformly for $x$ and $y$ in a compact subset
of $\R$. Observe also that $E_0(u)=\mathcal O (|a|^{3/4})$ as $a \to
-\infty$; see \eqref{eq: E0} and \eqref{eq:asy of Minfty 0}. The
same bound holds for $E_0(u)^{-1}$, $E_0(v)$ and $E_0(v)^{-1}$, so
that we find
\[
E_0(v)^{-1}E_0(u)=I+\mathcal O \left(
|u-v||a|^{3/2}\right)=I+\mathcal O \left(\frac{|x-y|}{|a|^{3/4}}
\right).
\]

Combining all these results, gives
\begin{multline*}
\lim_{a\to -\infty} \frac{1}{\sqrt 2 |a|^{1/4}}\Ktac
\left(\frac{x}{\sqrt 2 |a|^{1/4}},\frac{y}{\sqrt 2 |a|^{1/4}};-\tfrac12 a^2,|a|\left(1+\frac{\sigma}{2|a|^{3/2}}\right)\right) \\
=\frac{1}{2 \pi i (x-y)}  \begin{pmatrix} -1 & 0 & 1  \end{pmatrix} \Phi(y;\sigma)^{-1}
\Phi(x;\sigma)  \begin{pmatrix} 1\\0\\1 \end{pmatrix}.
\end{multline*}
An appeal to \eqref{eq:sol of modelRHP} yields
\begin{multline*}
\lim_{a\to -\infty} \frac{1}{\sqrt 2 |a|^{1/4}}\Ktac \left(\frac{x}{\sqrt 2 |a|^{1/4}},\frac{y} {\sqrt 2 |a|^{1/4}};-\tfrac12 a^2,|a|\left(1+\frac{\sigma}{2|a|^{3/2}}\right)\right)\\
=\frac{1}{2 \pi i (x-y)} \begin{pmatrix} 1 & 1 & 0  \end{pmatrix} \Phi^\Pe(y;\sigma)^{T}  \Phi^\Pe(x;\sigma)^{-T}  \begin{pmatrix}
1\\-1\\0 \end{pmatrix}.
\end{multline*}
Finally, by taking the transpose on both sides of the above formula,
\eqref{eq: tacnode th 2} follows from \eqref{eq: Pearcey
kernel} for the case $x,y>0$.

Next we prove \eqref{eq: tacnode th 1} from symmetry considerations. We start with the following symmetry relation
\[
M^{-T}(\zeta;s,t)=\begin{pmatrix} 0 & I_2 \\ -I_2 & 0 \end{pmatrix} M(\zeta;s,-t) \begin{pmatrix} 0 & -I_2 \\ I_2 & 0 \end{pmatrix},
\]
for $s,t \in \R$, which can be checked from RH problem \ref{rhp:
tacnode rhp}. Here, $I_2$ denotes $2\times 2$ identity matrix.
Using \eqref{eq:tacnode kernel} this leads to the following symmetry
property of the tacnode kernel
\[
\Ktac(u,v;s,-t)=\Ktac(v,u;s,t).
\]
Given this, \eqref{eq: tacnode th 1} is immediate from \eqref{eq: tacnode th 2}.

This completes the proof of Theorem \ref{th: tacnode2Pearcey}.

\subsection{Proof of Theorem \ref{th: Pearcey}}
\label{sec:proof of critical}

We start with a lemma that establishes a symmetry property of the
Pearcey parametrix.

\begin{lemma} \label{lemma: symmetry Pearcey}
Let $\Phi^{\Pe}(\cdot;\rho)$ be the unique solution of RH problem \ref{rhp: Pearcey}. Then
\begin{equation} \label{eq: symmetry Pearcey parametrix}
\diag \left( -i,1,i \right) \Phi^{\Pe}(iz;-\rho) B_j= \Phi^{\Pe}(z;\rho), \qquad \text{for $z$ in the $j$-th quadrant,}
\end{equation}
where
\[
B_1= \begin{pmatrix} 0&0&-1 \\ -1&0&0 \\ 0&1&0 \end{pmatrix}, \quad
B_2= \begin{pmatrix} 0&1&0 \\ -1&0&0 \\ 0&0&1  \end{pmatrix}, \quad
B_3= \begin{pmatrix} 0&1&0 \\ 0&0&1 \\ 1&0&0   \end{pmatrix}, \quad
B_4= \begin{pmatrix} 0&0&-1 \\ 0&1&0 \\ 1&0&0  \end{pmatrix}.
\]
\end{lemma}
\begin{proof}
We have to check that the left-hand side of \eqref{lemma: symmetry Pearcey} satisfies RH problem \ref{rhp: Pearcey}.
It is straightforward to check the jump conditions in item (2). Checking the asymptotics is a
bit more cumbersome. To that end the following observations are useful
\begin{align*}
\theta_1(iz;-\rho) &= \begin{cases} \theta_3(z;\rho), & \quad \text{for $z \in I \cup III \cup IV$}, \\
\theta_1(z;\rho), & \quad\text{for $z \in II$,} \end{cases} \\
\theta_2(iz;-\rho) &= \begin{cases} \theta_1(z;\rho), & \quad\text{for $z \in I \cup III \cup IV$}, \\
\theta_2(z;\rho), & \quad\text{for $z \in II$,} \end{cases} \\
\theta_3(iz;-\rho) &= \begin{cases} \theta_2(z;\rho), & \quad \text{for $z \in I \cup III \cup IV$}, \\
\theta_3(z;b), & \quad\text{for $z \in II$,} \end{cases}
\end{align*}
and also
\begin{align*}
e^{\pm \pi i/6}(iz)^{\mp 1/3} &= z^{\mp 1/3}, && \text{for $z \in I \cup III \cup IV$,} \\
\mp i(iz)^{\mp 1/3} &= z^{\mp 1/3}, && \text{for $z \in II$.}
\end{align*}
We omit the details here.
\end{proof}

The proof is similar to that of Theorem \ref{th: tacnode2Pearcey},
and we will focus on the points where both proofs differ and leave
some details to the reader.
Note that we are now working in the context of the two-matrix model, so
we choose the bottom sign in $\pm \sigma$; cf. \eqref{eq:gamma
explicit} and \eqref{t:scaling}.

Following the same ideas as in the proof of Theorem \ref{th:
tacnode2Pearcey}, we obtain
\begin{multline*}
a^2\Kcr \left(a^2u,a^2v;-\tfrac12 a^2,|a|\left(1-\frac{\sigma}{2|a|^{3/2}}\right)\right)=\frac{1}{2 \pi i (u-v)}  \\
\times \begin{pmatrix} -1 & 1 & 0 & 0 \end{pmatrix}\begin{pmatrix}
\Phi(f(iu);\rho(iu))^{-1} & 0 \\ 0&\phi(iu)^{-1}\end{pmatrix}
E_0(iu)^{-1}M^{(5)}_+(iu)^{-1} \\ \times
M^{(5)}_+(iv)E_0(iv)\begin{pmatrix} \Phi(f(iv);\rho(iv)) & 0 \\
0&\phi(iv)\end{pmatrix}
 \begin{pmatrix} 1\\1\\0\\0 \end{pmatrix},
\end{multline*}
for $0<|u|,|v|<\delta$.

Now we fix $x,y\neq0$ and take
\begin{equation*}
u=2^{-1/2}|a|^{-9/4}x,\qquad  v=2^{-1/2}|a|^{-9/4}y,
\end{equation*}
so that $|u|,|v| <\delta$ for $|a|$ sufficiently large. Hence
\begin{align*}
f(iu)\to ix, \qquad f(iv)\to iy,
\end{align*}
as $a \to -\infty$ and
\begin{align*}
\rho(iu)\to - \sigma, \qquad \rho(iv)\to - \sigma,
\end{align*}
as $a \to -\infty$. Furthermore,
\[
M^{(5)}_+(iu;a)^{-1} M^{(5)}_+(iv;a)=I+\mathcal O
\left(\frac{x-y}{|a|^{15/4}}\right),
\]
as $ a \to -\infty$, uniformly for $x$ and $y$ in a compact subset
of $\R$. Also
\[
E_0(iu)^{-1}E_0(iv)=I+\mathcal O \left( |u-v||a|^{3/2}\right)=I+\mathcal O \left(\frac{|x-y|}{|a|^{3/4}} \right).
\]

These results, together with \eqref{eq:sol of modelRHP} and taking
the transpose, imply
\begin{multline*}
\lim_{a\to -\infty} \frac{1}{\sqrt 2 |a|^{1/4}}\Kcr \left(\frac{x}{\sqrt 2 |a|^{1/4}},\frac{y} {\sqrt 2 |a|^{1/4}};-\tfrac12 a^2,|a|\left(1-\frac{\sigma}{2|a|^{3/2}}\right)\right)\\
=\frac{1}{2 \pi i (x-y)} \begin{pmatrix} 0 & 1 & 1  \end{pmatrix} \Phi^\Pe(iy;-\sigma)^{-1}  \Phi^\Pe(ix;-\sigma)  \begin{pmatrix}
0\\-1\\1 \end{pmatrix}.
\end{multline*}
Finally, we obtain Theorem \ref{th: Pearcey} by applying Lemma
\ref{lemma: symmetry Pearcey} and \eqref{eq: Pearcey kernel} to the
above formula.

\section*{Acknowledgments}
We thank Steven Delvaux and Arno Kuijlaars for their careful reading
of the manuscript and helpful comments. DG is a Research Assistant
of the Fund for Scientific Research - Flanders (FWO), Belgium. LZ is
a Postdoctoral Fellow of the Fund for Scientific Research - Flanders
(FWO), Belgium.

\end{document}